\documentclass[10pt,twocolumn,twoside]{IEEEtran}
\usepackage{graphics} % for pdf, bitmapped graphics files
\usepackage{epsfig} % for postscript graphics files
\usepackage{amsmath} % assumes amsmath package installed
\usepackage{amssymb}  
\usepackage{cite}
\usepackage{color}
\usepackage{graphicx}
\usepackage[caption=false]{subfig}
\usepackage{algpseudocode}
\usepackage{algorithm}

\renewcommand{\theenumi}{(\alph{enumi})}

\usepackage[font=footnotesize]{caption}

\newtheorem{theorem}{Theorem}[section]
\newtheorem{definition}[theorem]{Definition}
\newtheorem{assumption}[theorem]{Assumption}
\newtheorem{lemma}[theorem]{Lemma}

\newtheorem{remark}[theorem]{Remark}
\newtheorem{example}[theorem]{Example}

\newtheorem{proposition}[theorem]{Proposition}

\newcommand{\naturals}{\mathbb{N}}
\newcommand{\real}{\mathbb{R}}
\newcommand{\cplx}{\mathbb{C}}

\newcommand{\range}{\mathcal{R}}

\newcommand{\Ac}{\mathcal{A}}

\newcommand{\Dc}{\mathcal{D}}

\newcommand{\Fc}{\mathcal{F}}

\newcommand{\Kc}{\mathcal{K}}
\newcommand{\Lc}{\mathcal{L}}
\newcommand{\Mc}{\mathcal{M}}
\newcommand{\Nc}{\mathcal{N}}

\newcommand{\Pc}{\mathcal{P}}

\newcommand{\Sc}{\mathcal{S}}

\newcommand{\until}[1]{\{1,\dots,#1\}}

\newcommand{\dx}{D(X)}
\newcommand{\dy}{D(Y)}
\newcommand{\dxs}{D(X_s)}
\newcommand{\dys}{D(Y_s)}
\newcommand{\dxn}[1]{D(X_{#1})}
\newcommand{\dyn}[1]{D(Y_{#1})}
\newcommand{\tdx}{\tilde{D}(X)}
\newcommand{\tdy}{\tilde{D}(Y)}
\newcommand{\Span}{\operatorname{span}}

\newcommand{\rows}{{\operatorname{rows}}}
\newcommand{\cols}{{\operatorname{cols}}}
\newcommand{\rown}{{\operatorname{\sharp rows}}}
\newcommand{\coln}{{\operatorname{\sharp cols}}}
\newcommand{\basis}{{\operatorname{basis}}}
\newcommand{\dist}{{\operatorname{dist}}}
\newcommand{\ssd}{{\operatorname{SSD}}}
\newcommand{\cssd}{C_{\ssd}}
\newcommand{\Kssd}{K_\ssd}
\newcommand{\inn}{\Nc_{\operatorname{in}}} 
\newcommand{\outn}{\Nc_{\operatorname{out}}}
\newcommand{\flag}{{\operatorname{flag}}}

\newcommand{\new}[1]{{\color{blue} #1}}
\renewcommand{\new}[1]{#1}

\newcommand{\longthmtitle}[1]{\mbox{}{\textit{(#1):}}}
\newcommand{\setdef}[2]{\{#1 \; | \; #2\}}

\newcommand{\oprocendsymbol}{\hbox{$\square$}}
\newcommand{\oprocend}{\relax\ifmmode\else\unskip\hfill\fi\oprocendsymbol}

\renewcommand{\baselinestretch}{.99}

\allowdisplaybreaks

\title{\LARGE \bf Parallel Learning of Koopman Eigenfunctions and
  Invariant Subspaces For Accurate Long-Term Prediction\thanks{This
    work was supported by ONR Award N00014-18-1-2828.}\thanks{A
    preliminary version appeared as~\cite{MH-JC:20-acc} at the
    American Control Conference.}}

\author{Masih Haseli and Jorge Cort\'es% <-this % stops a space
  \thanks{Masih Haseli and Jorge Cort\'es are with Department of
    Mechanical and Aerospace Engineering, University of California,
    San Diego, CA 92093, USA, {\tt\small
      \{mhaseli,cortes\}@ucsd.edu}}%}% <-this % stops a space
}

\graphicspath{{epsfiles/}}

\begin{document}

\maketitle

\begin{abstract}
  We present a parallel data-driven strategy to identify
  finite-dimensional functional spaces invariant under the Koopman
  operator associated to an unknown dynamical system.  We build on the
  Symmetric Subspace Decomposition (SSD) algorithm, a centralized
  method that under mild conditions on data sampling provably
  finds the maximal Koopman-invariant subspace and all Koopman
  eigenfunctions in an arbitrary finite-dimensional functional space.
  A network of processors, each aware of a common dictionary of
  functions and equipped with a local set of data snapshots,
  repeatedly interact over a directed communication graph. Each
  processor receives its neighbors' estimates of the invariant
  dictionary and refines its estimate by applying SSD with its local
  data on the intersection of the subspaces spanned by its own
  dictionary and the neighbors' dictionaries.  We identify conditions
  on the network topology to ensure the algorithm identifies the
  maximal Koopman-invariant subspace in the span of the original
  dictionary, characterize its time, computational, and communication
  complexity, and establish its robustness against communication
  failures.
  % and packet~drops.
\end{abstract}

\section{Introduction}\label{sec:introduction}
Advances in processing and computation have led to numerous
opportunities to enhance our understanding of complex dynamic
behaviors. These opportunities, in turn, have challenged researchers
to explore alternative representations of dynamical systems that take
advantage of novel technologies.  The Koopman operator is one such
representation: rather than reasoning over trajectories, the operator
describes the effect of the dynamics on functions defined over the
state space.  The operator is appealing for data-driven identification
of dynamical phenomena since it is linear, independently of the
structure of the underlying dynamics. This explains the significant
insights into the physical laws governing the system provided by its
eigenfunctions.  Linearity is an advantage over other data-driven
learning methods such as neural networks and statistical methods,
which often lead to highly nonlinear models.  However, the
infinite-dimensional nature of the Koopman operator prevents the use
of efficient algorithms developed for digital computers. One can
circumvent this by finding finite-dimensional functional spaces that
are invariant under the Koopman operator.
% Since the action of the operator on such subspaces is exact, it can
% be encapsulated in a matrix product, which leads to fast
% linear-algebraic methods for accurate representation of the
% (unknown) \new{dynamical features}, with key implications for model
% reduction, stability, identification of stable/unstable manifolds,
% and control design.
The action of the operator on such subspaces is exact, 
  %	can be encapsulated in a matrix product. This
  enabling the use of fast linear algebraic methods with key
  implications for model reduction, stability, identification of
  stable/unstable manifolds, and control design.

\textit{Literature Review:} The Koopman operator is a linear operator
that represents the effect of a potentially nonlinear dynamics on a
linear functional space~\cite{BOK:31,BOK-JVN:32}.  The
eigendecomposition of the operator is particularly well suited to
provide a comprehensive description of dynamical behavior since the
Koopman eigenfunctions have linear temporal
evolutions~\cite{IM:05,MB-RM-IM:12}. As an example,~\cite{AM-IM:16}
provides several criteria for global asymptotic stability of
attractors based on Koopman eigenfunctions and their corresponding
eigenvalues.  Approximations of the Koopman operator and its
eigendecomposition enable identification of nonlinear dynamical
systems~\cite{AM-JG:16} and model reduction for nonlinear
systems~\cite{SK-FN-PK-HW-IK-CS-FN,AA-JNK:17}, which leads to easier
analysis of complex systems such as fluid
flows~\cite{CWR-IM-SB-PS-DSH:09}. The
  works~\cite{MK-IM-automatica:18,SP-SK:17,BH-XM-UV:18,EK-JNK-SLB:17,DG-DAP:17,GM-MC-XT-TM:19,AN-JSK:20,SHS-AN-JSK:20}
  provide data-driven control schemes based on the identification of
  finite-dimensional linear representations for the Koopman operator
  associated with nonlinear systems. The
  works~\cite{AS-JE-AR-JPC-RMD:19,AM-JB-MM:20} rely on symmetry to
  improve the accuracy of the Koopman approximation and characterize
  its observability properties.

% This has led to a wide variety of applications such as global
% stability analysis~\cite{AM-IM:16}, system
% identification~\cite{AM-JG:16,BK-PG-PB-SN-MB:17}, model
% reduction~\cite{SK-FN-PK-HW-IK-CS-FN,AA-JNK:17},
% robotics~\cite{GM-MC-XT-TM:19}, and
% control~\cite{MK-IM-automatica:18,SP-SK:17,BH-XM-UV:18,EK-JNK-SLB:17}.
%
To circumvent the challenges posed by the infinite-dimensional nature
of the Koopman operator,
% regarding the use of efficient methods compatible with digital
% computers, 
significant research efforts have been devoted to find
finite-dimensional approximations.
 % There are two main groups of data-driven strategies to find
 % finite-dimensional representations for the Koopman operator: (i)
 % projection-based methods and (ii) invariant subspace methods.
Dynamic Mode Decomposition (DMD)~\cite{PJS:10,JHT-CWR-DML-SLB-JNK:13}
uses linear algebraic methods to form a linear model based on the
observed data. Extended Dynamic Mode Decomposition (EDMD) is a
generalization of DMD that lifts the states of the systems to a space
spanned by a dictionary of functions and finds a linear model by
solving a data-driven least-squares
problem~\cite{MOW-IGK-CWR:15}. Unlike DMD, EDMD provides a way to
analyze the effect of the dynamics on an arbitrary functional space.
EDMD converges~\cite{MK-IM:18} to the projection of the Koopman
operator on the subspace spanned by the dictionary as the number of
sampled data goes to infinity. This indicates the importance of the
chosen dictionary in capturing important information about the
operator. If the dictionary spans a Koopman-invariant subspace, then
the EDMD approximation is exact; otherwise, EDMD loses information
about the system and might not suitable for long-term prediction due
to errors in the approximation. This observation motivates the search
for methods to identify Koopman-invariant subspaces.

%Dynamic Mode Decomposition (DMD)~\cite{PJS:10,JHT-CWR-DML-SLB-JNK:13}
%is one of the most popular method belonging to group (i). DMD uses
%linear algebraic methods to form a linear model based on the data.
%Extended Dynamic Mode Decomposition (EDMD) is a generalization of DMD
%that lifts the states of the systems to a linear functional space and
%finds a linear model by solving a data-driven least-squares
%problem~\cite{MOW-IGK-CWR:15}. The work~\cite{MK-IM:18} provides
%several results for convergence of EDMD to the Koopman operator.
%% Moreover, DMD and EDMD have been generalized to work with noisy
%% data~\cite{MSH-CWR-EAD-LNC:17,MH-JC:19-acc}.
%The approximation by EDMD is not generally suitable for long-term
%predictions since it merely converges to the projection of the Koopman
%operator on a finite-dimensional space and might lose information
%about the operator.
%
%One can circumvent this issue by finding finite-dimensional subspaces
%that are invariant under the application of the Koopman operator,
%since the linear model acquired by EDMD is exact on those
%subspaces. This is the subject of the second group of data-driven
%approaches. 
Since Koopman eigenfunctions automatically generate invariant
subspaces, a number of works have focused the attention on finding
them~\cite{SLB-BWB-JLP-JNK:16,EK-JNK-SLB:17}, including the use of
multi-step trajectory predictions~\cite{MK-IM:19} and sparsity
promoting methods~\cite{SP-NAM-KD:20} that prune the EDMD
eigendecomposition.  The latter relies on the fact that EDMD captures
all the Koopman eigenfunctions in the span of the dictionary,
cf.~\cite{MH-JC:20-tac}, but not all identified eigenfunctions using
EDMD are actually Koopman eigenfunctions. The literature contains
several methods based on neural networks that can approximate
  Koopman-invariant
  subspaces~\cite{QL-FD-EMB-IGK:17,NT-YK-TY:17,EY-SK-NH:19,SEO-CWR:19}.
None of the aforementioned methods provide guarantees for the
identified subspaces to be Koopman invariant. Our recent
work~\cite{MH-JC:20-tac} provides a necessary and sufficient condition
for the identification of functions that evolve linearly according to
the dynamics based on the application of EDMD forward and backward in
time, and establishes conditions on the density of data sampling to
ensure that the identified functions are Koopman eigenfunctions almost
surely. The work also introduces SSD algorithms which, should all data
be centrally available and under mild conditions on data
  sampling, provably find the maximal Koopman-invariant subspace in
the span of an arbitrary finite-dimensional functional space. By
  contrast, here we profit from parallelization to speed up
  computation and enable fast processing for large datasets and
  real-time applications.

\textit{Statement of Contributions:} We present a parallel data-driven
method to identify Koopman-invariant subspaces and eigenfunctions
associated with a (potentially nonlinear) unknown discrete-time
dynamical system. The proposed algorithm is compatible with parallel
processing hardware. Our starting point is a group of processors
communicating through a directed graph, with each processor aware of a
common dictionary of functions and equipped with a local set of data
snapshots acquired from the dynamics.  We introduce the Parallel
Symmetric Subspace Decomposition (P-SSD) algorithm to find the maximal
Koopman-invariant subspace and all the Koopman eigenfunctions in the
finite-dimensional linear functional space spanned by the dictionary.
The proposed strategy has each processor refine its estimate of the
invariant dictionary by iteratively employing the information received
from its neighbors to prune it.
% At each iteration, each agent receives the neighbors' estimates of the
% invariant dictionary, calculates the intersection of the subspaces
% spanned by those dictionaries and its own dictionary. Then, the agent
% applies the SSD algorithm on the calculated subspace with its local
% data. Finally, the agent updates its dictionary with a basis for the
% subspace calculated via SSD and transmits it to its neighbors. 
We show that the P-SSD algorithm reaches an equilibrium in a finite
number of time steps for any (possibly time-varying) network topology
and carefully characterize the properties of the agents' dictionary
iterates along its execution, particularly in what concerns
monotonicity of the associated subspaces. We also establish that the
globally reachable processors in the communication digraph find the
same solution as the SSD algorithm would if all data was centrally
available at a single processor.  This allows us to conclude that the
P-SSD algorithm finds the maximal Koopman-invariant subspace in the
space spanned by the original dictionary if the network topology is
strongly connected.  Finally, we conclude by characterizing the
algorithm's time, computational, and communication complexity,
demonstrating its computational advantage over SSD, and showing
its robustness against communication failures and packet drops.
Simulations illustrate the superior performance of the proposed
strategy over other methods in the literature\footnote{Throughout
    the paper, we use the following notation.  We use $\naturals$,
    $\naturals_0$, $\real$, and $\cplx$, to denote the sets of
    natural, nonnegative integer, real, and complex numbers,
    resp. Given integers $a,b$ we use $a \bmod b$ to represent
    the remainder of division of $a$ by~$b$. Given a matrix $A \in
    \cplx^{m \times n}$, we represent its set of rows, set of columns,
    number of rows, and number of columns by $\rows(A)$, $\cols(A)$,
    $\rown(A)$, and $\coln(A)$, resp. In addition, we denote
    its transpose, conjugate transpose, pseudo-inverse, and range
    space by $A^T$, $A^H$, $A^\dagger$, and $\range(A)$. If $A$ is
    square, we denote its inverse by $A^{-1}$. Given $A \in \cplx^{m
      \times n}$ and $B \in \cplx ^{m \times d}$, we use $[A,B] \in
    \cplx^{m \times (n+d)}$ to represent the matrix created by
    concatenating $A$ and $B$. Given $v \in \cplx^n$, we denote its
    2-norm by $\|v\|_2 := \sqrt{v^H v}$. Given vectors $v_1,\ldots,
    v_k \in \cplx^n$, we denote by $\Span\{v_1,\ldots, v_k \}$, the
    set comprised of all vectors in the form of $c_1v_1+\cdots+c_nv_n$
    with $c_1,\dots, c_n \in \cplx$. $\angle(v,w)$ represents the
    angle between $v,w \in \real^n$. For sets $S_1$ and $S_2$, $S_1
    \subseteq S_2$ means that $S_1$ is a subset of $S_2$. In addition,
    we denote the intersection and union of $S_1$ and $S_2$ by $S_1
    \cap S_2$ and $S_1 \cup S_2$, resp. Given functions $f:S_2
    \to S_1$ and $g:S_3 \to S_2$, we denote their composition by $f
    \circ g: S_3 \to S_1$.}.

\section{Preliminaries}\label{sec:preliminaries}
In this section, we present a brief account of Koopman operator theory
and basic definitions from graph theory.

\subsubsection*{Koopman Operator}
Our exposition here follows~\cite{MB-RM-IM:12}. Consider the
discrete-time dynamical system
\begin{align}\label{eq:dymamical-sys}
  x^+ = T(x),
\end{align}
where $T: \Mc \to \Mc$ is defined over the state space $\Mc \subseteq
\real^n$. 
The Koopman operator associated
with~\eqref{eq:dymamical-sys} characterizes the effect of the dynamics
on functions (also known as observables) in a linear functional space
$\Fc$ defined from $\Mc$ to $\cplx$. If $\Fc$ is closed under
composition with $T$ (i.e., $ f \circ T \in \Fc$, for all $f \in
\Fc$), one can define the Koopman operator $\Kc: \Fc \to \Fc$
associated with~\eqref{eq:dymamical-sys} as
\begin{align*}
  \Kc(f) = f \circ T.
\end{align*}
Since $\Fc$ is a linear space, the Koopman operator is \emph{spatially
  linear}, i.e.,
\begin{align}\label{eq:Koopman-spatial-linear}
  \Kc (c_1 f_1+ c_2 f_2) = c_1 \Kc(f_1) + c_2 \Kc(f_2),
\end{align}
for every $f_1, f_2 \in \Fc$ and $c_1,c_2 \in \cplx$.
Unlike~\eqref{eq:dymamical-sys}, which describes the effect of the
dynamics on points in the state space, the Koopman operator describes
the effect of the dynamics on functions in $\Fc$. If $\Fc$ contains
functions that can describe the states of the system, the Koopman
operator fully describes the evolution of the states in time. One way
to enforce this richness criteria is to include the state observables
$f_i(x) := x_i$ in the space $\Fc$ for all $i \in \until{n}$. This
condition, together with its closedness under $T$, might force~$\Fc$
to be infinite dimensional.

Having a linear operator enables us to define the eigendecomposition
of the Koopman operator. Formally, the function $\phi \in \Fc$ is an
eigenfunction of the Koopman operator associated with eigenvalue
$\lambda \in \cplx$ if
\begin{align}\label{eq:Koopman-eigendecomposition}
  \Kc(\phi) = \lambda \phi.
\end{align}
The eigenfunctions of the Koopman operator evolve linearly in time,
i.e., $ \phi(x^+) =(\phi \circ T)(x) = \Kc(\phi)(x) = \lambda
\phi(x)$.  This \emph{temporal linearity} in conjunction with the
\emph{spatial linearity} in~\eqref{eq:Koopman-spatial-linear} render
the spectral properties of the Koopman operator a powerful tool in
analyzing nonlinear dynamical systems. Formally, given a function
\begin{align}\label{eq:function-modes}
  f = \sum_{i=1}^{N_k} c_i \phi_i,
\end{align}
where $\{\phi_i\}_{i=1}^{N_k}$ are Koopman eigenfunctions with
eigenvalues $\{\lambda_i\}_{i=1}^{N_k}$ and $\{c_i\}_{i=1}^{N_k}
\subset \cplx$, one can use the spectral properties of the Koopman
operator and describe the temporal evolution of $f$ according to the
dynamics as
\begin{align}\label{eq:function-evolution-Koopman}
  f(x(k)) = \sum_{i=1}^{N_k} c_i \lambda_i^k \, \phi_i (x(0)), \quad
  \forall k \in \naturals.
\end{align}
It is important to note that due to the infinite-dimensional nature of
the Koopman operator, one might need infinitely many eigenfunctions to
describe an arbitrary function in a linear manner. One way to avoid
this infinite-dimensionality problem is to analyze the functions in
finite-dimensional subspaces that are invariant under the application
of the Koopman operator. A subspace $\Sc \subseteq \Fc$ is
\emph{Koopman invariant} if $\Kc(f) \in \Sc$ for all $f \in
\Sc$. Moreover, $\Sc \subseteq \Pc$ is the \emph{maximal
  Koopman-invariant} subspace of $\Pc$ if $\Sc$ is Koopman invariant
and, for every Koopman-invariant subspace $\Lc \subseteq \Pc$, we have
$\Lc \subseteq \Sc$.

\subsubsection*{Graph Theory}
Our exposition here mainly
follows~\cite{FB-JC-SM:08cor,JL-ASM-AN-TB:17}. Given a set $V$
comprised of $m$ \emph{nodes} and a set $E \subseteq V \times V$
comprised of ordered pairs of nodes called \emph{edges}, the pair
$G=(V,E)$ defines a \emph{directed graph (digraph)} on nodes $V$ and
edges $E$. We say node $i$ is an \emph{in-neighbor} of node $j$ and
node $j$ is an \emph{out-neighbor} of node $i$ if $(i,j) \in E$.  We
denote by $\inn(i)$ and $\outn(i)$ the sets of in- and out-neighbors
of~$i$, resp.  A \emph{directed path} with \emph{length} $l$
is an ordered sequence of $l+1$ nodes such that the ordered pair of
every two consecutive nodes is an edge of the digraph.  A path is
\emph{closed} if its first and last nodes are the same. A node is
called \emph{globally reachable} if there is a directed path from
every other node to it. A digraph is \emph{strongly connected} if all
of its nodes are globally reachable. Given two nodes $i,j$ in a
digraph, the \emph{distance} from $i$ to $j$, denoted $\dist(i,j)$, is
the length of the shortest directed path from $i$ to $j$. If there is
no such directed path, the distance is $\infty$.  Given two digraphs
$G_1 =(V,E_1)$ and $G_2 = (V,E_2)$, we define their \emph{composition}
as $G_2 \circ G_1 = (V,E_\circ)$ such that $E_\circ =
\setdef{(i,j)}{\exists k \in V \text{ with } (i,k)\in E_1 \land (k,j)
  \in E_2}$. A finite sequence of digraphs $\{G_i = (V,E_i)\}_{i=1}^k$
is \emph{jointly strongly connected} if $G_k \circ \cdots \circ G_1$
is strongly connected. An infinite sequence of digraphs (or a
time-varying digraph) $\{G_i = (V,E_i)\}_{i=1}^\infty$ is
\emph{repeatedly jointly strongly connected} if there exists $l, \tau
\in \naturals$ such that for every $k \in \naturals_0$, the sequence
of digraphs $\{ G_{\tau +kl +j}\}_{j=0}^{l-1}$ is jointly strongly
connected.
% For a time-varying digraph, we denote by $\inn^k(i)$ and
% $\outn^k(i)$ the set comprised of in-neighbors and out-neighbors of
% node $i$ at time $k$ resp.

\section{Problem Setup}\label{sec:problem-statement}
Given an arbitrary finite-dimensional functional space $\Pc$ defined
on a state space $\Mc$, our goal is to design efficient data-driven
algorithms that are able to identify the maximal Koopman-invariant
subspace in $\Pc$ associated with a dynamical system whose explicit
model is unknown. To do so, we rely on data collected about the
dynamical behavior of the system and aim to ensure that the proposed
methods are compatible with parallel processing hardware.  

We start by defining the concepts of data snapshots and dictionary. We
represent by $X, Y \in \real^{N \times n}$ matrices comprised of $N$
data snapshots acquired from system~\eqref{eq:dymamical-sys} such that
$x_i^T$ and $y_i^T$, the $i$th rows of $X$ and $Y$, satisfy
\begin{align*}
  y_i = T(x_i), \quad \forall i \in \until{N}.
\end{align*}
Moreover, we define a dictionary $D: \Mc \to \real^{1 \times N_d}$
\begin{align*}
  D(x) = [d_1(x), \ldots, d_{N_d}(x) ],
\end{align*}	
comprised of $N_d$ real-valued functions $d_1, \ldots, d_{N_d}$
defined on $\Mc$ with $\Span \{d_1, \ldots, d_{N_d} \} = \Pc$. Note
that, one can completely characterize any dictionary $\tilde{D}$ with
elements in $\Span\{d_1, \ldots, d_{N_d}\}$ by a matrix $C$ with $N_d$
rows as $\tilde{D}(x) = D(x) C$.  The effect of the dictionary on data
matrices~is
\begin{align*}
  D(X) = [D(x_1)^T,\ldots,D(x_N)^T]^T \in \real^{N \times N_d}.
\end{align*}

Throughout the paper, we operate under the following standard
assumption on dictionary matrices.

\begin{assumption}\longthmtitle{Full Column Rank Dictionary
    Matrices}\label{a:full-rank}
  The matrices $\dx$ and $\dy$ have full column rank.  \oprocend
\end{assumption}

Assumption~\ref{a:full-rank} requires the dictionary functions to be
linearly independent and the data snapshots to be diverse enough to
capture the important features of the dynamics.
% data snapshots capturing important characteristics of the dynamics.

We consider a group of $M$ processors or agents that communicate
according to a digraph~$G$.  Every agent is aware of the
dictionary. We assume the data snapshots $X$, $Y$ are distributed
among the processors, forming the dictionary snapshots $\dxn{i},
\dyn{i}$ for each agent $i \in \until{M}$ such that
\begin{align*}
  \bigcup_{i=1}^M \rows([\dxn{i}, \dyn{i}]) &= \rows([\dx, \dy]).
\end{align*}
Local data snapshots might be available at the agent as a result of
distributed acquisition or because the global data has been
partitioned among the processors. The latter can be done in different
ways, including partitioning in a way that minimizes the maximum time
taken for each agent to finish one iteration of calculation. For
homogeneous processors, this can be done by uploading the signature
data sets to agents and evenly distributing the rest of the data among
them.

The processors share a small portion of dictionary snapshots, $\dxs$
and $\dys$ with full column rank, called the \emph{signature
  dictionary matrices} (note that one can build signature dictionary
snapshots with no more than $2N_d$ rows according to
Assumption~\ref{a:full-rank}).  Hence, we have
\begin{align*}
    \rows([\dxs, \dys]) &\subseteq \rows([\dxn{i}, \dyn{i}]),
\end{align*}
for $i \in \until{M}$. Interestingly, for our forthcoming
results, the processors do not need to be aware of which data points
correspond to the signature snapshots.

\emph{Problem Statement:} Given a group of $M$ processors
communicating according to a network, the dictionary
$D=[d_1,\dots,d_{N_d}]$, and the data snapshots $X,Y \in \real^{N
  \times n}$, we aim to design a parallel algorithm that is able to
identify the dictionary $\tilde{D}$ spanning the maximal
Koopman-invariant subspace in $\Span\{d_1,\dots,d_{N_d}\}$. The
processors must only rely on the their local data and communication
with their neighbors to achieve this goal.

In general, the action of Koopman operator on the functional space
$\Pc = \Span\{d_1,\dots,d_{N_d}\}$ is not invariant, meaning that one
needs to project back to $\Pc$ the image under $\Kc$ of functions in
$\Pc$, thereby introducing error. Instead, the image of functions
belonging to the identified invariant subspace under the Koopman
operator remains in the subspace, avoiding the need for the projection
and therefore eliminating the source of approximation error. We
  note that efficient solutions to this problem could be employed in
  the construction of better dictionaries. Although beyond the scope
  of this paper, one could envision, for instance, augmenting the
  dictionary to increase its linear span of functions and then pruning
  it to eliminate the approximation error.  Throughout the paper, we
use data-driven methods that are not specifically designed to work
with data corrupted with measurement noise. Hence, one might
need to pre-process the data before applying the proposed algorithms.

\section{Parallel Symmetric Subspace Decomposition}

This section presents a parallel algorithm to identify the dictionary
$\tilde{D}$ spanning the maximal Koopman-invariant subspace in
$\Span\{d_1,\dots,d_{N_d}\}$, as stated in
Section~\ref{sec:problem-statement}.  We start by noting that the
invariance of $\tilde{D}$ implies that $ \range(\tilde{D}(x^+)) =
\range(\tilde{D}(x))$, for all $ x \in \Mc$.  This gets reflected in
the data as $\range(\tilde{D}(Y)) = \range(\tilde{D}(X))$.
% \begin{align*}
%  \range(\tilde{D}(Y)) = \range(\tilde{D}(X)).
%\end{align*}
Using the fact that $\tilde{D}$ is fully characterized by a matrix $C$
with $\tilde{D}(x) = D(x) C$, this equation can be rewritten as 
\begin{align}\label{eq:range-old-dictionary-ssd}
  \range(D(Y)C) = \range(D(X)C).
\end{align}
Hence, under appropriate conditions on the data sampling, the
problem of finding the maximal Koopman-invariant subspace can be
formulated in algebraic terms by looking for the full column rank
matrix $C$ with maximum number of columns such
that~\eqref{eq:range-old-dictionary-ssd} holds.

\begin{algorithm}
  \caption{Symmetric Subspace Decomposition} \label{algo:ssd}
  \begin{algorithmic}[1] % The number tells where the line numbering should start
  	\Statex \textbf{Inputs:} $\dx, \dy \in \real^{N \times
            N_d }$ \qquad \textbf{Output:} $\cssd$
    \State \textbf{Initialization} 
    \State $i \gets 1$, $A_1 \gets \dx$, $B_1 \gets \dy$, $\cssd \gets I_{N_d}$
    % \State $n_1 \gets 2$, $m_1 \gets 1$ \Comment{Arbitrary with $m_1 < n_1$}
    % \While{$m_i < n_i$}\smallskip
    \While{1}
    \smallskip
    \State $\begin{bmatrix} Z_i^A \\ Z_i^B \end{bmatrix} \gets
    \operatorname{null}([A_i,B_i])$ \Comment{Basis for the null
      space}\smallskip \label{algossd:null}
    \If{$\operatorname{null}([A_i,B_i])=\emptyset$}
    \State \textbf{return} $0$ \Comment{The basis does not exist}
    \State \textbf{break}
    \EndIf
    % \State $[n_i^A, m_i^A] \gets \operatorname{size}(Z_i^A)$
    % \State $n_i^A \gets$ number of rows of $Z_i^A$
    % \State $m_i^A \gets$ number of  columns of $Z_i^A$
    \If{$\rown(Z_i^A)  \leq \coln(Z_i^A)$} \label{algossd:size-check}
    \State \textbf{return} $\cssd$ \Comment{The procedure is complete}
    \State \textbf{break}
    \EndIf
    \State $\cssd \gets \cssd Z_i^A$ \Comment{Reduce the subspace}
    \State $A_{i+1} \gets A_i Z_i^A$, $B_{i+1} \gets B_i Z_i^A$, $i \gets i+1$
    % \State $B_{i+1} \gets B_i Z_i^A$ 
    % \State $i \gets i+1$
    \EndWhile
  \end{algorithmic}
\end{algorithm}

When all the data is available at a single processor, the Symmetric
Subspace Decomposition (SSD) algorithm proposed
in~\cite{MH-JC:20-tac}, cf. Algorithm~\ref{algo:ssd}, is a centralized
procedure that identifies $\tilde{D}$.
% by iteratively pruning the original dictionary to
% enforce~\eqref{eq:range-old-dictionary-ssd}.  (for reference,
% Appendix~\ref{app:ssd} summarizes the main properties of SSD).
Note that~\eqref{eq:range-old-dictionary-ssd} implies
\begin{align*}
  \range(\dx C) = \range(\dy C) \subseteq \range(\dx) \cap \range(\dy) ,
\end{align*}	
which in turn leads to rank deficiency in the concatenated matrix
$[\dx,\dy]$. SSD prunes at each iteration the dictionary using this
rank deficiency until~\eqref{eq:range-old-dictionary-ssd} holds for
the final dictionary $\tilde{D}$.  For reference,
Appendix~\ref{app:ssd} summarizes the main properties of SSD. Here, we
just point out that, under appropriate conditions on the data
sampling, $\tilde{D}$ is a basis for the maximal Koopman-invariant
subspace in the span of the original dictionary almost surely.  The
computational cost of SSD grows as $N$ and $N_d$ grow, mainly due to
Step~\ref{algossd:null}, which requires calculation of the null space
of a $N \times (2N_d)$ matrix. 

Here, we build on the SSD algorithm to address the problem laid
  out in Section~\ref{sec:problem-statement}, i.e.,
  identifying~$\tilde{D}$ by means of a parallel strategy, either
  because the data is not centrally available at a single processor
  or, even if it is, because its implementation over multiple
  processors speeds up the identification significantly.  Our
algorithm, termed Parallel Symmetric Subspace
Decomposition\footnote{The function $\basis(\Ac)$ returns a matrix
  whose columns provide a basis for $\Ac$. If $\Ac =\{0\}$, then
  $\basis(\Ac)$ returns $0$. Moreover, $\coln(0)=0.$} (P-SSD),
cf. Algorithm~\ref{algo:pssd}, has each processor use SSD on its local
data and prune its subspace using the estimates of the invariant
dictionary communicated by its neighbors.  The following is an
informal description of the algorithm steps:
\begin{quote}
  \emph{Informal description of P-SSD:} Given dictionary snapshots
  $\dx$, $\dy$ distributed over $M$ processors communicating over a
  network,
  % first creates small full rank signature dictionary snapshots
  % $\dxs$, $\dys$ from $\dx$, $\dy$ and uploads them to the
  % processors. Then, it distributes the rest of the data among the
  % processors. This distribution can be done in many different ways
  % but it is helpful to distribute the data in a such way that
  % enhances the execution of the algorithm. For example, one can
  % distribute the data on homogeneous processors as evenly as
  % possible to minimize the maximum time taken by the processors to
  % execute one iteration of the algorithm. After the data
  % distribution, each processor executes the procedure described by
  % Steps~\ref{al:p-start}-\ref{al:p-finish} of
  % Algorithm~\ref{algo:pssd}. In that procedure,
  each agent iteratively: (i) receives its neighbors' estimates of the
  invariant dictionary, (ii) uses the SSD algorithm to identify the
  largest invariant subspace (according to the local data) in the
  intersection of its dictionary and its neighbors' dictionaries,
  (iii) chooses a basis for the identified subspace as its own
  dictionary, and (iv) transmits that dictionary to its neighbors.
\end{quote}
The proposed P-SSD algorithm builds on the observation that the
subspaces identified by SSD are monotone non-increasing with respect to
the addition of data (Lemma~\ref{l:ssd-monotone}), i.e., the dimension
of the identified subspace does not increase when we add more data.

% \marginJC{I've removed the initialization from the algo, since the way
%   we've stated the problem in Section III already takes care of the
%   data partitioning.}
%
\begin{algorithm}[htb]
  \caption{Parallel Symmetric Subspace Decomposition}\label{algo:pssd}
  \begin{algorithmic}[1]
    % The number tells where the line numbering should start
    % \Statex \textbf{Initialization} 
    % % 
    % \State Create full column rank signature snapshots:
    % % 
    % \Statex $\dxs, \dys$ 
    % % 
    % \State Distribute the rest of data among agents and form: 
    % % 
    % \Statex  $\dxn{i},\dyn{i}, \forall i \in \until{M}$ \smallskip
    % % 
    % \State Upload the signature snapshots to the agents:
    % % 
    % \Statex $\dxn{i} \gets [\dxs^T,\dxn{i}^T]^T$  \smallskip
    % % 
    % \Statex $\dyn{i} \gets [\dys^T,\dyn{i}^T]^T$
    % % 
    % \Statex \textbf{End Initialization} 
    % % 
    % \smallskip
    % % 
    \Statex \textbf{Inputs:} $\dxn{i}, \dyn{i}$, $i \in \until{M}$
    \Statex \textbf{Each agent $\mathbf{i}$ executes:}
    \State $k \gets 0$, $C_0^i \gets I_{N_d}$, $\flag_0^i \gets 0$ \label{al:p-start}
    \While{1}
    \State $k \gets k+1$
    \State Receive $C_{k-1}^j, \, \forall j \in \inn^k(i)$
    \State $D_k^i \gets \basis \Big(\bigcap_{j \in \{i\} \cup
      \inn^k(i)} \range(C_{k-1}^j)\Big)$\smallskip \label{al:basis}
    \State $E_k^i \gets
    \ssd(\dxn{i}D_k^i,\dyn{i}D_k^i)$ \label{al:ssd}
    \If{$\coln(D_k^i E_k^i) < \coln(C_{k-1}^i)$} \label{al:if}
    \State $C_k^i \gets D_k^i E_k^i$ \Comment{The subspace gets
      smaller}  \label{al:newsub}
    \State $\flag_k^i \gets 0$   \label{al:flag0}
    \Else
    \State $C_k^i \gets C_{k-1}^i$ \Comment{The subspace does not
      change}  \label{al:oldsub}
    \State $\flag_k^i \gets 1$   \label{al:flag}
    \EndIf
    \State Transmit $C_k^i$ to out-neighbors
    \EndWhile \label{al:p-finish}
  \end{algorithmic}
\end{algorithm}

\subsection{Equilibria and Termination of P-SSD}

Here, we define the concept of equilibria of the P-SSD algorithm and
discuss how to detect whether the agents have attained one (we refer
to this as termination).  Viewing the algorithm as a discrete-time
dynamical system, we start by defining the concept of equilibrium.

\begin{definition}\longthmtitle{Equilibrium}
  The agent matrices $C_k^i$, $i \in \until{M}$ in
  Algorithm~\ref{algo:pssd} are an equilibrium of P-SSD at time $k \in
  \naturals$ if $ C_{p}^i = C_k^i$ for all $p>k$ and all $i \in \until{M}$.
\end{definition}
\smallskip

Note that, after reaching an equilibrium, the subspaces identified by
the agents do not change anymore.  The next result shows that the
P-SSD algorithm always reaches an equilibrium in a finite number of
iterations.

\begin{proposition}\longthmtitle{Reaching Equilibrium in a Finite
    Number of Iterations}\label{p:pssd-consensus}
%  Under Assumption~\ref{a:full-rank}, the
  The P-SSD algorithm executed over any (possibly time-varying)
  digraph reaches an equilibrium after a finite number of iterations.
\end{proposition}
%\vspace*{2pt}
%%
%For space reasons, the proof is presented in the online
%  version~\cite{MH-JC:20-tcns}.
%%%%%%%%%%%%%%
% Removed for space reasons
%%%%%%%%%%%%%%
\begin{proof}
  We prove the result by contradiction. Suppose that the algorithm
  never reaches an equilibrium, i.e., there exists an increasing
  sequence $\{k_p\}_{p=1}^\infty \subseteq \naturals_0$ such that, at
  iteration $k_p$, at least an agent $j \in \until{M}$ executes
  Steps~\ref{al:if}-\ref{al:flag0} in Algorithm~\ref{algo:pssd}.
  Hence, $\coln(C_{k_p+1}^j)<\coln(C_{k_p}^j)$. Consequently, using
  the fact that the number of columns for the matrix of each agent is
  non-increasing in time (since they execute either
  Step~\ref{al:newsub} or Step~\ref{al:oldsub} of
  Algorithm~\ref{algo:pssd}) and the fact that $k_{p+1} \geq k_p+1$,
  one can write
  \begin{align*}
    \sum_{i=1}^{M} \coln(C_{k_{p+1}}^i) < \sum_{i=1}^{M}
    \coln(C_{k_p}^i), \, \forall p \in \naturals.
  \end{align*}
  Since $\sum_{i=1}^{M} \coln(C_{0}^i) = MN_d$, this equation implies
  that, for $p>MN_d+1$, $\sum_{i=1}^{M} \coln(C_{k_p}^i)<0$, a
  contradiction.
\end{proof}

Since the algorithm runs in parallel, we need a mechanism to detect if
the P-SSD algorithm has reached an equilibrium.  The flag variables
carry out this task.

\begin{proposition}\longthmtitle{Equilibria is Detected by Flag
    Variables in Time Invariant Digraphs}\label{p:pssd-flags}
  Given a time-invariant network and under the P-SSD algorithm,
  $\flag_l^i = 1$ for some $l \in \naturals$ and for every $i \in
  \until{M}$ if and only if $C_k^i = C_{l-1}^i$ and $\flag_k^i = 1$
  for every $i \in \until{M}$ and every $k \geq l$.
\end{proposition}
%\vspace*{2pt}
%%
% For space reasons, the proof is presented in the
%  online version~\cite{MH-JC:20-tcns}.
%%%%%%%%%%%%%%
% Removed for space reasons
%%%%%%%%%%%%%%
\begin{proof}
  The implication from right to left is straightforward (simply put
  $k=l$). To prove the other implication, we use strong induction. By
  hypothesis, $\flag_l^i = 1$ for every $i \in \until{M}$ and the
    P-SSD algorithm has executed Steps~\ref{al:oldsub}-\ref{al:flag}
    at iteration $l$. Hence, $C_l^i = C_{l-1}^i$ for all $i \in
    \until{M}$. Now suppose that for every $i \in \until{M}$
  \begin{align}\label{eq:flag-induction}
    \flag_k^i = 1 , \, C_k^i = C_{l-1}^i, \, \forall k \in \{l+1 , l+2, \dots, p\}.
  \end{align}
  We need to prove that for every $i \in \until{M}$
  \begin{align*}
    \flag_{p+1}^i = 1 , \, C_{p+1}^i = C_{l-1}^i.
  \end{align*}
  Noting that the network is time-invariant, based on
  Step~\ref{al:basis} of Algorithm~\ref{algo:pssd} at iteration $p+1$
  and the fact that $C_p^i = C_{p-1}^i = C_{l-1}^i$ for every $i \in
  \until{M}$, we have $D_{p+1}^i = D_p^i = D_{l-1}^i$ for every $i \in
  \until{M}$. Similarly, $E_{p+1}^i = E_p^i = E_{l-1}^i$ for every $i
  \in \until{M}$. As a result,
  \begin{align}\label{eq:cols-equal-1}
    \coln(D_{p+1}^i E_{p+1}^i) = \coln(D_p^i E_p^i),
  \end{align}
  for all $ i \in \until{M}$.  Based on the algorithm
  and~\eqref{eq:flag-induction},
  \begin{align}\label{eq:cols-equal-2}
    \coln(D_p^i E_p^i) \geq \coln(C_{p-1}^i) = \coln(C_{l-1}^i), 
  \end{align}
  for all $ i \in \until{M}$.
  Using~\eqref{eq:cols-equal-1},~\eqref{eq:cols-equal-2}, and the fact
  that $\coln(C_p^i) = \coln(C_{l-1}^i)$, we have
  \begin{align*}
    \coln(D_{p+1}^i E_{p+1}^i) \geq \coln(C_p^i) ,
  \end{align*}
  for all $i \in \until{M}$.  Consequently, the algorithm executes
  Steps~\ref{al:oldsub}-\ref{al:flag} for every agent and we have
  \begin{align*}
    \flag_{p+1}^i = 1 , \, C_{p+1}^i = C_{l-1}^i,
  \end{align*}
  for all $ i \in \until{M}$, which concludes the proof.
\end{proof}

Note that the flags detect an equilibrium with one time-step delay.
We say that the P-SSD algorithm has \emph{terminated at step $k$} if
$\flag_k^i = 1$ for all $i \in \until{M}$.

\begin{remark}\longthmtitle{Termination of the P-SSD
    Algorithm}\label{r:flags-termination}
  We point out that one may consider different notions of termination
  for P-SSD: (i) the algorithm reaches an equilibrium, i.e., the
  agents continue their calculations but do not change their outputs;
  (ii) based on Proposition~\ref{p:pssd-flags}, a user external to the
  parallel processing hardware terminates the algorithm when
  $\prod_{i=1}^{M} \flag_k^i =1$ for some $k \in \naturals$; (iii)
  agents use distributed halting algorithms~\cite{NAL:97,DP:00} to
  decide when to terminate based on their output and flags. Here, we
  only employ (i) and~(ii) as appropriate and do not implement~(iii)
  for space reasons.  \oprocend
\end{remark}

%\begin{remark}\longthmtitle{Algorithm Termination via 
%    Flag Variables}%\label{r:flags-termination}
%    Under the specifications of the P-SSD strategy detailed in
%    Algorithm~\ref{algo:pssd}, agents continue their calculations
%    after reaching an equilibrium; however, the calculated matrices
%    do not change. The flag variables provide a way for the user to
%    be aware of the occurrence of an equilibrium.  Hence, if
%    $\prod_{i=1}^{M} \flag_k^i =1$ for some $k \in \naturals$, the
%    user can deduce, based on Proposition~\ref{p:pssd-flags}, that
%    the algorithm has reached an equilibrium and terminate the
%    process.  \oprocend
%\end{remark}

\subsection{Properties of Agents' Matrix Iterates along P-SSD}
Here we characterize important properties of the matrices computed by
the agents at each iteration of the P-SSD algorithm. The results
provided in this section hold for any (including time-varying)
communication network topology.  The next result characterizes basic
properties of the agents' matrix iterates.

\begin{theorem}\longthmtitle{Properties of Agents'
    Matrix Iterates}\label{t:pssd-agent-prop}
  Under the P-SSD algorithm and for any (possibly time-varying)
  digraph, for each $i \in \until{M}$,
  \begin{enumerate}
  \item for all $k \in \naturals_0$, $C_k^i$ is zero or has full column rank;
  \item for all $k \in \naturals$, $\range(D_k^i E_k^i) = \range(C_k^i)$;
  \item for all $k \in \naturals$, $\range(\dxn{i} C_k^i) = \range(\dyn{i} C_k^i)$.
  \end{enumerate}
\end{theorem}
\begin{proof}
  Consider an arbitrary $i \in \until{M}$.  We prove (a) by induction.
  For $k=0$, $C_0^i = I_{N_d}$. Now, suppose that $C_k^i$ is zero or
  has full column rank, and let us show the same fact for $k+1$.  Note
  that, if $D_{k+1}^i = 0$ or $E_{k+1}^i = 0$, then $C_{k+1}^i = 0$
  and the result follows. Now, suppose that $D_{k+1}^i \neq 0$ and
  $E_{k+1}^i \neq 0$. Based on definition of $\basis$ and
  Theorem~\ref{t:SSD-convergence}(a), $D_{k+1}^i$ and $E_{k+1}^i$ have
  full column rank. Note that the algorithm either executes
  Step~\ref{al:newsub} or executes Step~\ref{al:oldsub}. Hence,
  $C_{k+1}^i = D_{k+1}^i E_{k+1}^i $ or $C_{k+1}^i = C_{k}^i
  $. Consequently, one can deduce that $C_{k+1}^i$ has full column
  rank ($C_{k}^i \neq 0$ since $D_{k+1}^i \neq 0$, hence $C_{k}^i$ has
  full column rank), establishing (a).
	
  Next, we show (b). If the algorithm executes Step~\ref{al:newsub} at
  iteration $k$, then $C_k^i = D_k^i E_k^i$ and consequently
  $\range(D_k^i E_k^i) = \range(C_k^i)$.  Suppose instead that the
  algorithm executes Step~\ref{al:oldsub}, and consequently $C_k^i =
  C_{k-1}^i$. For the case $C_k^i = 0$, using
  Step~\ref{al:basis}, we deduce $D_k^i E_k^i = 0$ and the
  statement follows.  Assume then that $C_k^i \neq 0$ and hence it has
  full column rank as a result of~(a). Consequently, since the
    condition in Step~\ref{al:if} does not hold and using the
    definition of $\coln$, we have $D_k^i \neq 0$ and $E_k^i \neq 0$.
  Also, they have full column rank based on the definition of $\basis$
  and Theorem~\ref{t:SSD-convergence}(a). Moreover, $\range(D_k^i
  E_k^i) \subseteq \range(D_k^i) \subseteq \range(C_{k-1}^i)$ and
  since the matrices have full column rank, $\coln(D_k^i E_k^i) \leq
  \coln(D_k^i ) \leq \coln(C_{k-1}^i)$.  In addition, and based on
  Step~\ref{al:if}, $\coln(D_k^i E_k^i) \geq
  \coln(C_{k-1}^i)$.~Hence,
  \begin{align}\label{eq:coln-equality-1}
    \coln(D_k^i E_k^i) = \coln(D_k^i) = \coln(C_{k-1}^i).
  \end{align}
  We can use $\rown(E_k^i) = \coln(D_k^i)$ and $\coln(D_k^i E_k^i) =
  \coln(E_k^i)$ in conjunction with~\eqref{eq:coln-equality-1} to
  deduce that $E_k^i$ is square and nonsingular (since it has full
  column rank). Consequently, $\range(D_k^i) = \range(D_k^i E_k^i)$.
%  \begin{align*}%\label{eq:range-equality-1}
%    \range(D_k^i) = \range(D_k^i E_k^i).
%  \end{align*}
  This fact, in combination with $\range(D_k^i) \subseteq
  \range(C_{k-1}^i)$, yields
  \begin{align}\label{eq:range-equality-2}
    \range(D_k^i E_k^i) \subseteq \range(C_{k-1}^i).
  \end{align}
  Moreover, since $D_k^i$, $E_k^i$, and $C_{k-1}^i$ have full column
  rank, one can use~\eqref{eq:coln-equality-1}
  and~\eqref{eq:range-equality-2} to deduce that $\range(D_k^i E_k^i)
  = \range(C_{k-1}^i)$. Finally, since $C_k^i = C_{k-1}^i$, we have $
  \range(D_k^i E_k^i) = \range(C_k^i)$, which shows (b).
  	
  To show (c), from Step~\ref{al:ssd} of the algorithm and
  Theorem~\ref{t:SSD-convergence}(a), for each $i \in \until{M}$ and
  each $k \in \naturals$,
  \begin{align*}
    \range(\dxn{i} D_k^i E_k^i) = \range(\dyn{i} D_k^i E_k^i).
  \end{align*}
  Based on part~(b) and the fact that $\dxn{i},\dyn{i}$ have full
  column rank, one can deduce $ \range(\dxn{i} C_k^i) = \range(\dyn{i}
  C_k^i)$, concluding the proof.
\end{proof}

\begin{remark}\longthmtitle{Computation of the $\basis$ Function with
    Finite-Precision
    Machines}\label{r:basis-intersection-implementation}
  Based on Theorem~\ref{t:pssd-agent-prop}(a), we provide here a
  practical way of implementing the $\basis$ function. Since the
  matrices $C_k^i$ are either full rank or zero, instead of working
  with their range space, one can work directly with the matrices
  themselves. Here, we compute iteratively the output of the basis
  function for the input matrices. Given the full column-rank matrices
  $A_1,A_2$, we find $\basis(\range(A_1) \cap \range(A_2))$
  using~\cite[Lemma A.1]{MH-JC:20-tac}
 % Lemma~\ref{I:subspace-intersection}(a)
  by finding the null space of $[A_1,A_2]$ (if one of the matrices is
  equal to zero, then $\basis(\range(A_1) \cap \range(A_2)) =0$). This
  procedure can be implemented iteratively for $i \geq
  2$ by noting
  \begin{align*}
    &\basis \Big( \bigcap_{j=1}^{i+1} \range(A_i) \Big)
    \\
    &= \basis \bigg(\range\Big( \basis \Big( \bigcap_{j=1}^{i}
    \range(A_i) \Big) \Big) \bigcap \range(A_{i+1})\bigg). 
  \end{align*}
  When implemented on digital computers, this might lead to small
  errors affecting the null space of $[A_1,A_2]$.  To avoid this
  problem, given the singular values $\sigma_1 \geq \cdots \geq
  \sigma_{l}$ of $[A_1,A_2]$ and a tuning parameter
  $\epsilon_{\cap} >0$, we compute $k$ as the minimum integer
  satisfying $\sum_{j=k}^{l} \sigma_j^2 \leq \epsilon_{\cap}
    ( \sum_{j=1}^{l} \sigma_j^2 )$
  %
%  \marginMH{I moved the equation inline and changed its shape.}
  %
%  \begin{align*}
%    \frac{\sum_{j=k}^{l} \sigma_j^2}{\sum_{j=1}^{l} \sigma_j^2 }
%    \leq \new{\epsilon_{\cap}},
%  \end{align*}
  and then set $\sigma_k = \cdots = \sigma_{l}=0$.  The parameter
  $\epsilon_{\cap}$ tunes the sensitivity of the implementation.  \oprocend 
\end{remark}

Next, we show that the range space identified by an agent at any time
step is contained into the range space identified by that agent and
its neighbors in the previous time step.

\begin{lemma}\longthmtitle{Subspace Monotonicity of Agents'
    Matrix Iterates}\label{l:pssd-sub-monotone}
  Under the P-SSD algorithm and for any (possibly time-varying)
  digraph, at each iteration $k \in \naturals$, it holds that
  $\range(C_k^i) \subseteq \range(C_{k-1}^j)$ for all $i \in
  \until{M}$ and all $j \in \{i\} \cup \inn^k(i)$.
\end{lemma}
%\vspace*{2pt}
%%
%\new{For space reasons, the proof is presented in the online
%version~\cite{MH-JC:20-tcns}.}
%%%%%%%%%%%%%%
% Removed for space reasons
%%%%%%%%%%%%%%
\begin{proof}
  Let $i \in \until{M}$. The case $C_k^i = 0$ is trivial, so consider
  instead the case $C_k^i \neq 0$. This implies that $D_k^i \neq 0$
  and $E_k^i \neq 0$, with both having full column rank based on
  definition of the $\basis$ function and
  Theorem~\ref{t:SSD-convergence}(a).
  %	In the case that the algorithm executes Step~\ref{al:newsub}
  % we have $\range(C_k^i) \subseteq \range(C_{k-1}^j)$ since
  % $\range(D_k^i) \subseteq \range(C_{k-1}^i)$ and $\range(D_k^i
  % E_k^i) \subseteq \range(D_k^i)$.
  Note that $\range(D_k^i E_k^i) \subseteq \range(D_k^i)$ and, by
  definition, $\range(D_k^i) \subseteq \range(C_{k-1}^j)$ for every $j
  \in \{i\} \cup \inn^k(i)$. The result now follows by noting that
  $\range(D_k^i E_k^i) = \range(C_k^i)$ (cf.
  Theorem~\ref{t:pssd-agent-prop}(b)).
\end{proof}

We conclude this section by showing that the range space of the
agents' matrix iterates contain the \emph{SSD subspace}, i.e., the
subspace identified by the SSD algorithm when all the data is
available at once at a single processor.

\begin{proposition}\longthmtitle{Inclusion of SSD Subspace by Agents'
    Matrix Iterates}\label{p:inclusion-ssd-subspace}
  Let $ \cssd = \ssd (\dx,\dy)$ be the output of the SSD algorithm
  applied on $\dx, \dy$.  Under the P-SSD algorithm and for any (possibly
  time-varying) digraph, at each iteration $k \in \naturals_0$, it
  holds that $\range(\cssd) \subseteq \range(C_k^i)$ for all $i \in
  \until{M}$.
\end{proposition}
%\vspace*{2pt}
%%
%\new{For space reasons, the proof is presented in the online
%version~\cite{MH-JC:20-tcns}.}
%%%%%%%%%%%%%%
% Removed for space reasons
%%%%%%%%%%%%%%
\begin{proof}
  The case $\cssd = 0$ is trivial. Suppose that $\cssd \neq 0$. We
  prove the argument by induction. Since $C_0^i= I_{N_d}$ for all $i
  \in \until{M}$, we have
  \begin{align*}
    \range(\cssd) \subseteq \range(C_0^i) = \real^{N_d}, \, \forall i
    \in \until{M}.
  \end{align*}
  Now, suppose that 
  \begin{align}\label{eq:ssd-sub-intersection}
    \range(\cssd) \subseteq \range(C_k^i), \, \forall i \in \until{M} ,
  \end{align}
  and let us prove the same inclusion for $k+1$.  Based on
  Step~\ref{al:basis} in Algorithm~\ref{algo:pssd}, one can write
  \begin{align}\label{eq:basis-inclusion}
    \range(D_{k+1}^i )= \bigcap_{j \in \{i\} \cup \inn^{k+1}(i)}
    \range(C_k^j), \, \forall i \in \until{M}.
  \end{align}
  Using~\eqref{eq:ssd-sub-intersection}
  and~\eqref{eq:basis-inclusion}, we get $\range(\cssd) \subseteq
  \range(D_{k+1}^i )$ for every $i \in \until{M}$. Moreover, since
  both $D_{k+1}^i$ and $\cssd$ have full column rank (the first, based
  on its definition and the second based on
  Theorem~\ref{t:SSD-convergence}(a)), one can find matrices
  $F_{k+1}^i, i \in \until{M}$ with full column rank such that
  \begin{align}\label{eq:F-mat}
    \cssd = D_{k+1}^i F_{k+1}^i, \, \forall i \in \until{M}.
  \end{align}
  From Theorem~\ref{t:SSD-convergence}(a), we have 
  $\range(\dx\cssd) = \range(\dy \cssd)$. Hence, since
  $\rows([\dxn{i}, \dyn{i}]) \subseteq \rows([\dx, \dy])$ for every $i
  \in \until{M}$, we deduce
  \begin{align}\label{eq:SSD-satisfies-agent-data}
    \range(\dxn{i} \cssd) = \range(\dyn{i} \cssd).
  \end{align}
  % as a result of Assumption~\ref{a:full-rank} and the fact that the
  % data snapshots $\dxn{i} , \dyn{i}$ for $i \in \until{M}$ are
  % subsets of $\dx , \dy$.
  % \marginMH{Do I need to prove the line above?}
  Combining~\eqref{eq:F-mat} and~\eqref{eq:SSD-satisfies-agent-data},
  we obtain
  \begin{align}\label{eq:subspace-equal-new-dictionary}
    \range(\dxn{i} D_{k+1}^i F_{k+1}^i) = \range(\dyn{i} D_{k+1}^i
    F_{k+1}^i),
  \end{align}
  for every $i \in \until{M}$. Now, by Step~\ref{al:ssd} of
  Algorithm~\ref{algo:pssd}, for every $i \in \until{M}$,
  \begin{align*}
    E_{k+1}^i = \ssd (\dxn{i} D_{k+1}^i , \dyn{i} D_{k+1}^i ).
  \end{align*}
  Using Theorem~\ref{t:SSD-convergence}(b) with the dictionary
  $D(x)D_{k+1}^i$ and data $\dxn{i}$, $\dyn{i}$, we deduce
  \begin{align*}%\label{eq:subset-dictionary}
    \range(F_{k+1}^i) \subseteq \range(E_{k+1}^i), \forall i \in \until{M}.
  \end{align*}
  Combining this with~\cite[Lemma A.2]{MH-JC:20-tac}
  % Lemma~\ref{l:product-subspace}
  and~\eqref{eq:F-mat}, we get
  \begin{align}\label{eq:cssd-subset}
    \range(\cssd)= \range(D_{k+1}^i F_{k+1}^i) \subseteq
    \range(D_{k+1}^i E_{k+1}^i),
  \end{align}
  for every $i \in \until{M}$. According to Algorithm~\ref{algo:pssd},
  either $C_{k+1}^i = C_k^i$ (Step~\ref{al:oldsub}) or $C_{k+1}^i
  =D_{k+1}^i E_{k+1}^i$
  (Step~\ref{al:newsub}). Using~\eqref{eq:ssd-sub-intersection} in the
  former case and~\eqref{eq:cssd-subset} in the latter, we conclude $
  \range(\cssd) \subseteq \range(C_{k+1}^i)$, for all $ i \in
  \until{M}$.
\end{proof}

\section{Equivalence of P-SSD and SSD}\label{sec:equivalence-pssd-ssd}
In this section, we study under what conditions on the digraph's
connectivity, the P-SSD algorithm is equivalent to the SSD algorithm,
i.e., it finds the maximal Koopman-invariant subspace contained in the
span of the original dictionary. 
% \subsection{SSD Subspaces on Directed Paths in P-SSD}
% We start by studying how, under the P-SSD algorithm, the SSD subspace
% is found along directed paths of the digraph.
% The next result shows that the agents identify the symmetric
% subspaces corresponding to the incoming dictionary matrices along
% directed paths.
We start by studying the relationship between the matrix iterates and
local data of the agents at the beginning and end of a directed path in
the digraph.

\begin{proposition}\longthmtitle{Relationship of Agents' Matrices
      and Local Data Along Directed 
      Paths}\label{p:symmetric-sub-directed-path} 
    Given an arbitrary constant digraph, let $P$ be a directed path of
    length $l$ from node~$j$ to node~$i$. Then, under the P-SSD
    algorithm, for each iteration $k \in \naturals$ and all $q \geq
    k+l$,
  \begin{subequations}
    \begin{align}
      \range(C_{q}^i) &\subseteq \range
      (C_k^{j}), \label{eq:subspace-inclusion-directed-path}
      \\
      \range(\dxn{j}C_q^i) &=
      \range(\dyn{j}C_q^i). \label{eq:symmetric-subspace-directed-path}
    \end{align}
  \end{subequations}
\end{proposition}
\begin{proof}
  To prove~\eqref{eq:subspace-inclusion-directed-path}, we label each
  node on the path $P$ by $p_1 $ through $p_{l+1}$, with $p_1$
  corresponding to node $j$ and $p_{l+1}$ corresponding to node $i$
  (note that a node will have more than one label if it appears more
  than once in the path, which does not affect the argument). Based on
  Lemma~\ref{l:pssd-sub-monotone}, for all $k \in \naturals$, one can
  write
  \begin{align*}
    \range(C_{k+m}^{p_{m+1}}) \subseteq \range(C_{k+m-1}^{p_m}), \,
    \forall m \in \until{l}.
  \end{align*}
  Using this equation $l$ times yields $ \range(C_{k+l}^i) \subseteq
  \range (C_k^{j})$.  Moreover, using Lemma~\ref{l:pssd-sub-monotone}
  once again for node $i$, we deduce $\range(C_q^i) \subseteq
  \range(C_{k+l}^i)$ for every $q \geq k+l$, and the proof follows.

  Regarding~\eqref{eq:symmetric-subspace-directed-path}, the case
  $C_q^i = 0$ is trivial. Suppose then that $C_q^i \neq 0$, with full
  column rank
  (cf. Theorem~\ref{t:pssd-agent-prop}(a)). Using~\eqref{eq:subspace-inclusion-directed-path}
  and Theorem~\ref{t:pssd-agent-prop}(a), we deduce that $C_k^{j}$
  also has full column rank. Therefore, there exists a full column
  rank matrix $F$ such that
  \begin{align}\label{eq:matrix-product-directed-path}
    C_q^i = C_k^{j} F.
  \end{align}
  From Theorem~\ref{t:pssd-agent-prop}(c), we have $ \range(\dxn{i}
  C_q^i) = \range(\dyn{i} C_q^i)$.  Looking at this equation in a
  row-wise manner and considering only the signature matrices, one can
  write 
  \begin{align}\label{eq:signature-range-equality1}
    \range(\dxs C_k^{j} F) = \range(\dys C_k^{j} F),
  \end{align}
  where we have used~\eqref{eq:matrix-product-directed-path}.  From
  Theorem~\ref{t:pssd-agent-prop}(c) for agent $j$,
  \begin{align*}
    \range(\dxn{j} C_k^{j}) = \range(\dyn{j} C_k^{j}).
  \end{align*}
  Using this equation and given the fact that $\dxn{j}$, $\dyn{j}$,
  and $C_k^{j}$ have full column rank, there must exist a square
  nonsingular matrix $K$ such that
  \begin{align}\label{eq:matrix-equality1}
    \dyn{j} C_k^{j} = \dxn{j} C_k^{j} K.
  \end{align}
  Looking at this equation in a row-wise manner and considering only
  the signature data matrices, 
  \begin{align}\label{eq:signature-equality1}
    \dys C_k^{j} = \dxs C_k^{j} K.
  \end{align}
  Using a similar argument for~\eqref{eq:signature-range-equality1},
  one can deduce that there exists a nonsingular square matrix $K^*$
  such that
  \begin{align}\label{eq:auxx}
    \dys C_k^{j} F= \dxs C_k^{j} F K^*.
  \end{align}
  Multiplying both sides of~\eqref{eq:signature-equality1} from the
  right by $F$ and subtracting from~\eqref{eq:auxx} results in
  \begin{align*}
    \dxs C_k^{j} (F K^* - K F) = 0.
  \end{align*}
  Since $\dxs C_k^{j}$ has full column rank, we deduce $ F K^* = K F$.
  Now, by multiplying both sides of~\eqref{eq:matrix-equality1} from
  the right by $F$ and replacing $KF$ by $F K^*$, we can write
  \begin{align*}
    \dyn{j} C_q^i = \dyn{j} C_k^{j} F= \dxn{j} C_k^{j} F K^* = \dxn{j}
    C_q^i K^* ,
  \end{align*}
  where we have used~\eqref{eq:matrix-product-directed-path} twice,
  which shows the result.
\end{proof}

Next, we show that globally reachable nodes in the digraph determine the
SSD subspace in a finite number of iterations.

\begin{theorem}\longthmtitle{Globally Reachable Nodes Find the SSD
    Subspace}\label{t:globally-reachable-ssd-sub}
  Given an arbitrary constant digraph, let $i$ be a globally reachable
  node and define $l = \max_{j \in \until{M}} \dist(j,i)$. Then, under
  the P-SSD algorithm, $\range(C_k^i) = \range(\cssd) $ for all $k
  \geq l+1$.
\end{theorem}
\begin{proof}
  If $C_k^i = 0$ for some $k \in \naturals$, then based on
  Proposition~\ref{p:inclusion-ssd-subspace} and
  Theorem~\ref{t:SSD-convergence}(a), we have $\cssd = 0$ and
  consequently, $\range(C_k^i) = \range(\cssd)$. Now, suppose that for
  all $k \in \naturals$, $C_k^i \neq 0$ and has full column rank
  (cf. Theorem~\ref{t:pssd-agent-prop}(a)).  Since $i$ is a globally
  reachable node, for each node $j \in \until{M}\setminus \{i\}$,
  there exists a directed path with length $l_j = \dist(j,i) < \infty$
  to node $i$.  Using Proposition~\ref{p:symmetric-sub-directed-path}
  for $j \in \until{M}\setminus \{i\}$ and
  Theorem~\ref{t:pssd-agent-prop}(c) for agent $i$ , we can write
  \begin{align*}
    \range(\dxn{j} C_{l+1}^i) = \range (\dyn{j} C_{l+1}^i), \, \forall j \in \until{M}.
  \end{align*}
  Moreover, since for all $j \in \until{M}$, $\dxn{j}$, $\dyn{j}$, and
  $C_{l+1}^i$ have full column rank, there exist nonsingular square
  matrices $\{K_j\}_{j=1}^M$ such that
  \begin{align}\label{eq:signature-matrix-equality}
    \dyn{j} C_{l+1}^i = \dxn{j} C_{l+1}^i K_j, \, \forall j \in \until{M}.
  \end{align}
  Looking at~\eqref{eq:signature-matrix-equality} in a row-wise manner
  and only considering the signature data sets, one can write
  \begin{align*}
    \dys C_{l+1}^i = \dxs C_{l+1}^i K_j, \, \forall j \in \until{M}.
  \end{align*}
  For any $p \neq q \in \until{M}$, we
  subtract~\eqref{eq:signature-matrix-equality} evaluated at $j = p$
  and at $j = q$ to obtain $ \dxs C_{l+1}^i (K_q - K_p) = 0$. Since $\dxs$
  and $C_{l+1}^i$ have full column rank, this implies $K:= K_1 = \cdots =
  K_M $. Looking at~\eqref{eq:signature-matrix-equality} in a row-wise
  manner and considering the fact that
  \begin{align*}
    \bigcup_{j=1}^M \rows([\dxn{j}, \dyn{j}]) = \rows([\dx, \dy]),
  \end{align*}
  we deduce $\dx C_{l+1}^i K = \dy C_{l+1}^i$ and consequently
  \begin{align*}
    \range(\dx C_{l+1}^i) = \range(\dy C_{l+1}^i).
  \end{align*}
  This, together with Theorem~\ref{t:SSD-convergence}(b), implies
  $\range(C_{l+1}^i) \subseteq \range(\cssd)$. This inclusion along
  with Proposition~\ref{p:inclusion-ssd-subspace} yields
  $\range(C_{l+1}^i) = \range(\cssd)$. Finally, for every $k \geq
  l+1$, Lemma~\ref{l:pssd-sub-monotone} implies $\range(C_k^i)
  \subseteq \range(C_{l+1}^i) = \range(\cssd)$, which along with
  Proposition~\ref{p:inclusion-ssd-subspace} yields $\range(C_k^i) =
  \range(\cssd) $.
\end{proof}

% According to Theorem~\ref{t:globally-reachable-ssd-sub}, to identify
% the SSD subspace and consequently the maximal Koopman invariant
% subspace in the span of the dictionary. However, the user needs to
% be aware of the topology of the network to be able to identify the
% SSD subspace. Next, we develop methods that ease this requirement
% and are implementable on most parallel processing hardwares.

% \subsection{Consensus on the SSD Subspace over Strongly Connected
%   Networks}

% all agents reach a consensus on the subspace identified by the SSD
% algorithm given all the data.

Theorem~\ref{t:globally-reachable-ssd-sub} shows that the globally
reachable nodes identify the SSD subspace. Next, we use this
  fact to derive guarantees for agreement of the range space of all
  agents' matrices on the SSD subspace over strongly connected
  networks (we refer to this as P-SSD reaching consensus).

\begin{theorem}\longthmtitle{Consensus on SSD Subspace and Time
    Complexity of P-SSD over Strongly Connected
    Digraphs}\label{t:consensus-complexity}
  Given a strongly connected digraph with diameter $d$,
  % and under Assumption~\ref{a:full-rank},
  \begin{enumerate}
  \item P-SSD reaches consensus in at most $d+1$
    iterations: specifically, for all $k \geq d+1$ and all $i
    \in \until{M}$,
    \begin{align*}
      \range(C_k^i) &= \range(\cssd), \\
      \range(\dx C_k^i) &= \range(\dy C_k^i);
    \end{align*}
    
  \item the P-SSD algorithm terminates after at most $d+2$ iterations,
    i.e, $\flag_{d+2}^i = 1$ for each $i \in \until{M}$.
  \end{enumerate}
\end{theorem}
\begin{proof}
  Regarding (a), the proof of consensus on the SSD subspace readily
  follows from Theorem~\ref{t:globally-reachable-ssd-sub} by noting
  that any node is globally reachable from any other node through a
  path with at most $d$ edges. The rest of the statement is a
  corollary of this fact together with
  Theorem~\ref{t:SSD-convergence}(a).

  Regarding (b), one needs to establish that not only the range space
  of the agents' matrices remains the same but also that the P-SSD
  algorithm executes Steps~\ref{al:oldsub}-\ref{al:flag}, which means
  the matrices themselves also remain the same and flags become 1.
  Note that using~(a) we deduce that $ \range(C_{d+1}^i) =
  \range(C_{d+2}^i)$ for all $i \in \until{M}$. Hence, if $C_{d+1}^i =
  0$, then $C_{d+2}^i = 0$, and the P-SSD algorithm does not execute
  Steps~\ref{al:newsub}-\ref{al:flag0} but executes
  Steps~\ref{al:oldsub}-\ref{al:flag} instead. Consequently,
  $\flag_{d+2}^i = 1$ for every agent $i \in \until{M}$. Suppose then
  that $C_{d+1}^i \neq 0$, with full column rank based on
  Theorem~\ref{t:pssd-agent-prop}(a), and consequently $\cssd \neq 0$
  (also with full column rank, cf.
  Theorem~\ref{t:SSD-convergence}(a)). Using~(a), we then deduce that
  in Step~\ref{al:basis} of the P-SSD algorithm for agent $i$ at
  iteration $d+2$, we have $\range(D_{d+2}^i) = \range(\cssd)$.  Since
  $D_{d+2}^i$ has full column rank by definition,
  \begin{align*}
    \range(\dx D_{d+2}^i) = \range(\dy D_{d+2}^i).
  \end{align*}
  Looking at this equation in a row-wise manner and considering the
  fact that $ \rows([\dxn{i}, \dyn{i}]) \subseteq \rows ([\dx, \dy])$,
  \begin{align*}
    \range(\dxn{i} D_{d+2}^i I) = \range(\dyn{i} D_{d+2}^i I),
  \end{align*}
  where $I$ is the identity matrix with appropriate size. This fact,
  together with definition of $E_{d+2}^i$ in Step~\ref{al:ssd} of the
  P-SSD algorithm, implies that $\range(I) \subseteq
  \range(E_{d+2}^i)$
  (cf. Theorem~\ref{t:SSD-convergence}(b)). Consequently, $E_{d+2}^i$
  must be a nonsingular square matrix (cf.
  Theorem~\ref{t:SSD-convergence}(a)). Hence,
  \begin{align*}
    \coln(D_{d+2}^i E_{d+2}^i) = \coln(D_{d+2}^i) = \coln(\cssd) =
    \coln(C_{d+1}^i).
  \end{align*}
  Consequently, the P-SSD algorithm executes
  Steps~\ref{al:oldsub}-\ref{al:flag} and we have $C_{d+2}^i
  =C_{d+1}^i$ and $\flag_{d+2}^i = 1$ for every agent $i \in
  \until{M}$, concluding the proof.
\end{proof}

In general, the P-SSD algorithm might terminate faster than the
explicit upper bound given in Theorem~\ref{t:consensus-complexity}.
The main reason for this is that each agent depends on performing SSD
on its own data and \emph{usually} SSD can identify the maximal
Koopman-invariant subspace in the span of dictionary with a moderate
amount of data.  The next remark shows that the floating point
operation (FLOPs) complexity of the P-SSD algorithm is much lower for
each processor that the SSD algorithm and hence the P-SSD subspace
search runs much faster than SSD.

\begin{remark}\longthmtitle{Floating Point Operation (FLOP) Complexity of
    P-SSD}\label{r:flops-complexity}
  For the (standard) case that $N \gg N_d$, let $N_i$ be the number of
  data snapshots available to agent $i \in \until{M}$. Considering the
  fact that the most time consuming operation in
  Algorithm~\ref{algo:pssd} is Step~\ref{al:ssd}, one can
  use~\cite[Remark 5.2]{MH-JC:20-tac} and deduce that each iteration
  of the P-SSD algorithm takes $O(N_i N_d^3)$ FLOPs for agent $i$.
  Based on Theorem~\ref{t:consensus-complexity}(b), agent $i$ performs
  at most $O(N_i N_d^3 d)$ to find the SSD subspace. In case the data
  is uniformly distributed among the agents, for each $i \in
  \until{M}$,
  \begin{align*}
    N_i = O \Big(\rown(X_s) + \frac{N-\rown(X_s)}{M} \Big) = O(N/M),
  \end{align*}
  where in the last equality we have used $\rown(X_s) = O(N_d)$ (in
  fact one can find signature data with $\rown(X_s)\leq 2 N_d$ as a
  consequence of Assumption~\ref{a:full-rank}). Hence, the FLOPs
  complexity of each agent with data uniformly distributed across the
  agents is $O(dN N_d^3 /M)$. This gives a factor of $d/M$ reduction
  when compared to the FLOP complexity $O(N N_d^3)$ of SSD. If
    the processors are not homogeneous, then the uniform distribution
    of the data is not optimal. In general, the optimal distribution
    minimizes the maximum time taken for processors to complete one
    algorithm iteration, which means that faster processors should get
    more data. In practice, our simulations show that P-SSD runs
  drastically faster than the worst-case bound. \oprocend
\end{remark}

% Finally, we are ready to characterize the communication complexity of
% the P-SSD algorithm.

\begin{remark}\longthmtitle{Communication Complexity of P-SSD}
  Given a strongly connected digraph with $M$ nodes and diameter $d$,
  agent $i$ transmits $| \outn(i)|$ matrices with maximum size of $N_d
  \times N_d$. Thus, considering the real numbers as basic messages,
  and using the fact that the algorithm terminates after at most $d+1$
  iterations, the communication complexity of the P-SSD algorithm is
  \begin{align*}
    O(d N_d^2 \sum_{i=1}^M | \outn(i)|) = O(d N_d^2 E),
  \end{align*}
  where $E$ is the number of edges in the digraph.  In most conventional
  parallel computing units, the processors communicate through a
  shared bus. Hence, at iteration $k$ of the P-SSD algorithm, the
  $i$th agent only sends one message comprised of the matrix $C_k^i$
  to the bus. Therefore, over such networks, the communication
  complexity reduces to $O(d N_d^2 M)$.  \oprocend
\end{remark}

  \begin{remark}\longthmtitle{Approximated
      P-SSD}\label{r:approximated-PSSD}
    The original dictionary might not contain nontrivial
    Koopman-invariant subspaces. For such cases, one can look for
    \emph{approximate} informative invariant subspaces by replacing
    Step~\ref{al:ssd} of P-SSD with the Approximated-SSD
    algorithm~\cite[Section VII]{MH-JC:20-tac}, which approximates
    Koopman-invariant subspaces given a tuning parameter~$\epsilon$.
    This modification gives rise to the Approximated P-SSD algorithm
    for which a similar set of results to the ones stated here for the
    P-SSD algorithm regarding equilibria, finite termination and
    eventual agreement of the matrix iterates can be established (we
    omit them for space reasons).  The agents' identified subspaces
    satisfy the accuracy bounds provided in~\cite[Section
    VII]{MH-JC:20-tac} based on their local data.  \oprocend
\end{remark}

\section{Robustness Against Packet Drops and Time-Varying Networks} 
The parallel nature of P-SSD enables the use of parallel and
  distributed processing hardware, such as GPUs and clusters of
  processors.  The communication between processors required by such
  hardware is subject to packet drops or communication
  failures. Moreover, there are instances when some processors are
  needed for other tasks or go offline temporarily, resulting in
  time-varying networks. To address these
  issues, we investigate the robustness of the P-SSD algorithm
against packet drops in the communication network. To tackle this
problem, we model the network by a time-varying digraph where a
dropped packet from agent $i$ to agent $j$ at time $k \in \naturals$
corresponds to the absence of edge $(i,j)$ at time $k$ from the
digraph at that time. The next result shows that if the digraph
remains repeatedly jointly strongly connected, the agents executing
the P-SSD algorithm reach a consensus on the SSD subspace in finite
time.

\begin{theorem}\longthmtitle{Consensus on SSD Subspace and Finite-Time
    Convergence of P-SSD over Repeatedly Jointly Strongly Connected
    Digraphs}\label{t:consensus-jointly-strongly-connected}
  Given a time-varying repeatedly jointly strongly connected digraph
  $\{G_k=(\until{M},E_k)\}_{k=1}^\infty$, the P-SSD algorithm reaches
  a consensus equilibrium on the SSD subspace in finite time, i.e.,
  there exists $l \in \naturals$ such that for every iteration $k \geq
  l$, $C_k^i = C^i_{l}$ for all $i \in \until{M}$ with
  \begin{align*}
    \range(C_l^1)=\range(C_l^2) = \cdots = \range(C_l^M) =
    \range(\cssd).
  \end{align*}
\end{theorem}
\begin{proof}
  The fact that the P-SSD algorithm reaches an equilibrium is a direct
  consequence of Proposition~\ref{p:pssd-consensus}, which states that
  there exists $l \in \naturals$ such that
  \begin{align}\label{eq:sub-stays-the-same}
    C_k^i = C_l^i, \, \forall k \geq l, \, \forall i \in \until{M}.
  \end{align}
  Hence, we only need to show that the range spanned by these matrices
  corresponds to the SSD subspace.  Since the digraph is repeatedly
  jointly strongly connected, there exists a closed ``temporal'' path
  after time $l$ that goes through every node of the digraph. By this
  we mean that there exist $L$, time instants $l < k_1 < k_2 < \cdots
  < k_L$ and node labels $p_1, \dots , p_L$ covering all of $
  \until{M}$ (note that some nodes might correspond to more that one
  label) such that
  \begin{align*}
    (p_i,p_{(i \bmod L)+1}) \in E_{k_i}, \, \forall i \in \until{L}.
  \end{align*}
  Using~\eqref{eq:sub-stays-the-same}, Lemma~\ref{l:pssd-sub-monotone}
  at times $\{k_i\}_{i=1}^L$, and the fact that the path is closed,
  we deduce
  \begin{align*}
%    \range(C_l^{p_1}) \subseteq \cdots \subseteq \range(C_l^{p_L})
%    \subseteq \range(C_l^{p_1}).
%    \\
    \range(C_l^{p_1})  \subseteq \range(C_l^{p_L}) \subseteq \cdots \subseteq \range(C_l^{p_2})\subseteq \range(C_l^{p_1}).
  \end{align*}
  Hence, $ \range(C_l^{p_1}) = \cdots = \range(C_l^{p_L})$.  Since the
  path goes through every node and~\eqref{eq:sub-stays-the-same}
  again, we arrive at the consensus
  \begin{align}\label{eq:sub-equality}
    \range(C_k^1) = \cdots = \range(C_k^M), \, \forall k \geq l.
  \end{align}
  It remains to show that this consensus is achieved on the SSD
  subspace. The inclusion $ \range(\cssd) \subseteq \range(C_k^i) $
  for all $i \in \until{M}$ and $k \ge l$ follows from
  Proposition~\ref{p:inclusion-ssd-subspace}. To show the other
  inclusion, first note that if $C_k^i = 0$ for some $i \in \until{M}$
  and $k \geq l $, then $\cssd = 0$ and the proof follows.  Suppose
  then that $C_k^i$'s are nonzero with full column rank
  (cf. Theorem~\ref{t:pssd-agent-prop}(a)) for $k \geq l$. Based on
  Theorem~\ref{t:pssd-agent-prop}(c), for every $i \in \until{M}$ and
  $k \geq l$, we have $\range(\dxn{i} C_k^i) = \range(\dyn{i}
  C_k^i)$. Moreover, using~\eqref{eq:sub-equality} and the fact that
  all matrices $C_k^i$'s have full column rank, we have
  \begin{align*}%\label{eq:sub-equal-agents}
    \range(\dxn{i} C_k^j) = \range(\dyn{i} C_k^j), \, \forall
    i,j \in \until{M},
  \end{align*}
  for every $k \geq l$.  Using this equality for $j =1$ and $k = l$,
  \begin{align*} %\label{eq:sub-equal-one}
    \range(\dxn{i} C_l^1) = \range(\dyn{i} C_l^1), \, \forall i \in
    \until{M}.
  \end{align*}
  Hence, for every $i \in \until{M}$ there exists a nonsingular square
  matrix $K_i$ such that
  \begin{align}\label{eq:mat-equal-K}
    \dxn{i} C_l^1 K_i = \dyn{i} C_l^1 .
  \end{align}
  Now, looking at~\eqref{eq:mat-equal-K} in a row-wise manner, one
  obtains for the signature dictionary snapshots, $ \dxs C_l^1 K_i =
  \dys C_l^1$, for $ i \in \until{M}$, and hence
  \begin{align*}
    \dxs C_l^1 (K_i - K_j) = 0, \, \forall i,j \in \until{M}.
  \end{align*}
  Since $\dxs$ and $C_l^1$ have full column rank, we have $ K := K_1 =
  \cdots = K_M$.  Replacing $K_i$ by $K$ in~\eqref{eq:mat-equal-K}
  gives
  \begin{align*}%\label{eq:mat-equal}
    \dxn{i} C_l^1 K = \dyn{i} C_l^1, \, \forall i \in \until{M}.
  \end{align*}
  Looking at this equation in a row-wise manner and since
  \begin{align*}
    \bigcup_{i=1}^M \rows([\dxn{i}, \dyn{i}]) = \rows([\dx, \dy]),
  \end{align*}
  one can write $\dx C_l^1 K = \dy C_l^1$ and consequently
  \begin{align*}
    \range(\dx C_l^1) = \range(\dy C_l^1).
  \end{align*}
  Using now Theorem~\ref{t:SSD-convergence}(b), we deduce
  $\range(C_l^1) \subseteq \range(\cssd)$.
  Using~\eqref{eq:sub-stays-the-same} and~\eqref{eq:sub-equality}, we
  obtain that $\range(C_k^i) = \range(C_k^1) = \range(C_l^1) \subseteq
  \range(\cssd)$, for all $i \in \until{M}$ and $k \ge l$, and this
  concludes the proof.
\end{proof}

An interesting interpretation of the result in
Theorem~\ref{t:consensus-jointly-strongly-connected} is that, for
time-invariant networks subject to failures in the communication
links, as long as the agents re-connect at least once within some
uniform time period, the P-SSD algorithm identifies the SSD subspace
in a finite number of iterations.

\begin{remark}\longthmtitle{Using the Agents' Matrix Iterates to Find
    Maximal Koopman-Invariant Subspaces and
    Eigenfunctions}\label{r:inv-subs-eigs-pssd}
  After the agents reach consensus on the SSD subspace, they can use
  their computed matrix $C$ instead of $\cssd$ to find the dictionary
  $\tilde{D}(x) = \tilde{D}(x) C$, for $x \in \Mc$. Any agent $i \in
  \until{M}$ can use its P-SSD matrix to find
  \begin{align*}
    K_{\operatorname{P-SSD}} = \tilde{D}(X_i)^\dagger \tilde{D}(Y_i) =
    \tilde{D}(X_s)^\dagger \tilde{D}(Y_s).
  \end{align*}
  Note that this is also equal to the square matrix $\Kssd$ found by
  SSD, cf.~\eqref{eq:ssd-newdictionary-data}.  Importantly, all the
  results for SSD subspaces in~\cite{MH-JC:20-tac} are also valid for
  the new dictionary $\tilde{D}$. 
  For space reasons, we omit those results here and only mention that
  under some mild conditions on data sampling, $\tilde{D}(x)$ spans
  the maximal Koopman-invariant subspace in $\Span(D(x))$ and
  $\tilde{D}(x)w$ is an eigenfunction of the Koopman operator almost
  surely, given $K_{\operatorname{P-SSD}} w = \lambda w$ with $\lambda
  \in \cplx$ and $w \in \cplx^{\coln(\Kssd)} \setminus \{0\}$
  (see~\cite[Theorems 5.7-5.8]{MH-JC:20-tac} for more information).
  \oprocend
\end{remark} 

\section{Simulation Results}\label{sec:simulations}
%
%\marginMH{I removed the packet drop example. I think it would be
%  helpful to add it to the arXiv version and refer the reader to it in
%  appropriate places.}
%%
%\marginJC{I agree. Please remember to add it to the arxiv
%  version. I've added a sentence here referring the reader to it}
%
Here, we provide \new{four} examples
%\footnote{We refer the interested
%    reader to the online version~\cite{MH-JC:20-tcns} for an
%    additional example where we show the performance of P-SSD under
%    packet drops.} 
to demonstrate the properties and effectiveness of
the P-SSD algorithm%, and compare it with the SSD algorithm and
%Extended Dynamic Mode Decomposition (EDMD)
\footnote{We intentionally use
low-dimensional systems with sparse eigenfunctions to facilitate a
complete in-depth presentation of the results. However, we should
point out that the results are valid for high-dimensional systems with
no conditions on the sparsity of invariant subspaces.}.

\begin{example}\longthmtitle{Unstable Nonlinear
    System}\label{ex:nonlinear}
  Consider the discrete-time system with state $x = [x_1,x_2]$,
  \begin{subequations}\label{eq:ex-nonlinear}
    \begin{align}
      x_1^+ & = 1.2 \, x_1 
      \\
      x_2^+ & = \sqrt[3]{0.8 \, x_2^3 + 8 \, x_1^2 + 0.1} \, .
    \end{align}
  \end{subequations}
  The system can be transformed into a discrete-time
  polyflow~\cite{RMJ-PT:19}
  % with a state-inclusive polynomial Koopman-invariant subspace
  by the nonlinear transformation $[x_1,x_2] \mapsto [x_1,x_2^3]$. We
  aim to find the informative Koopman eigenfunctions and invariant
  subspaces.  % without transforming the system into a polyflow.
  we sample $N = 10^6$ data snapshots with initial conditions in the
  state space $\Mc = [-3,3] \times [-3,3]$ (dense sampling ensures
  that the properties identified here hold almost surely over the
  whole state space~\cite{MH-JC:20-tac}). We use the dictionary $D$
  with $N_d =15$ comprised of all distinct monomials up to degree 4 of
  the form $\prod_{i=1}^4 y_i$, with $y_i \in \{1,x_1,x_2\}$ for $i
  \in \until{4}$.
	
  To identify the maximal Koopman-invariant subspace in $\Span(D(x))$,
  we implement the SSD algorithm, and the P-SSD algorithm with $M \in
  \{5,20,100\}$ agents communicating according to a directed ring
  graph with diameter $d = M-1$. For the P-SSD strategy, we use the
  first 15 data snapshots in our database as signature data snapshots
  and distribute the rest of the data evenly among the agents.  Both
  strategies are implemented on a single computer using
  MATLAB\textsuperscript{\textregistered}.  We calculate the time
  elapsed for each iteration of P-SSD as the maximum time taken by
  agents to execute the algorithm in that iteration using the
  \texttt{tic} and \texttt{toc} commands.  Since we use finite
  precision, we apply the approximation provided in
  Remark~\ref{r:basis-intersection-implementation} with $\epsilon_\cap =
  10^{-12}$.

  For all $M \in \{5,20,100\}$, the P-SSD algorithm reaches the
  consensus equilibrium after 1 iteration and terminates (based on
  Remark~\ref{r:flags-termination}) after 2 iterations, which is
  significantly faster than the bounds provided in
  Theorem~\ref{t:consensus-complexity}. We have also observed in
  simulations that packet drops do not delay consensus.
  Both P-SSD and SSD algorithms correctly identify the 7-dimensional
  subspace spanned by $\{1, x_1, x_1^2, x_1^3, x_2^3, x_1^4, x_1
  x_2^3\}$ as the maximal Koopman-invariant subspace in $\Span(D(x))$.
  Table~\ref{table:elapsed-time} shows the time employed by
  the algorithms to find it. The P-SSD strategy with $M = 5$, $M =
  20$, and $M = 100$, is 80\%, 96\%, and 99\% faster than SSD,
  resp.

{ % begin box to localize effect of arraystretch change
\renewcommand{\arraystretch}{1.5}
\begin{table}[htb]
\centering
\caption{\color{black}Time elapsed to identify the maximal Koopman-invariant subspace in $\Span(D)$ associated with the dynamics~\eqref{eq:ex-nonlinear}.} \label{table:elapsed-time}
{\color{black}
\begin{tabular}[htb]{ | c || c | c | c | c | c | }
\hline
\textbf{Method}	              & SSD  		& P-SSD (5)  & P-SSD (20)  & P-SSD (100)   \\ \hline %
\textbf{Time (ms)}           & 2175        & 439     		& 88  				& 17         \\ \hline %
\end{tabular}
}
\end{table}
} % end box

%%%%%
% The following is removed for space reasons.
%%%%%
%  \begin{figure}[htb]
%    \centering 
%    \includegraphics[width=.70 \linewidth]{time-elapsed-different-methods-PSSD-SSD}%
%    \caption{Time elapsed to identify the maximal Koopman-invariant
%      subspace in $\Span(D)$ associated with the
%      dynamics~\eqref{eq:ex-nonlinear}.}\label{fig:elapsed-time-SSD-vs-PSSD}
%    \vspace*{-1.5ex}
%  \end{figure}
%  We also removed the signature snapshots from the agents and observed
%  that the P-SSD algorithm still correctly identifies the maximal
%  Koopman-invariant subspace. This is a consequence of the dense data
%  sampling employed here.

  Using the output matrix of any of the agents, cf.
  Remark~\ref{r:inv-subs-eigs-pssd}, we build the invariant dictionary
  $\tilde{D}$ and matrix
  $K_{\operatorname{P-SSD}}$. Table~\ref{table:eigenfunctions} shows
  the Koopman eigenfunctions and their corresponding eigenvalues
  calculated using the eigendecomposition of
  $K_{\operatorname{P-SSD}}$. One can verify analytically
  using~\eqref{eq:ex-nonlinear} that the eigenfunctions in
  Table~\ref{table:eigenfunctions} evolve linearly in time. Even
  though the function $x_2$ does not belong to the span of the
  Koopman-invariant subspace, the Koopman eigenfunctions fully capture
  the behavior of the system since the functions $x_1$ and $x_2^3$ do
  belong. Hence, one can use~\eqref{eq:function-evolution-Koopman} to
  predict the evolution of any function in form
  of~\eqref{eq:function-modes} or one simply can use
  \begin{align*}
    \tilde{D}(x^+) = \tilde{D}(x) K_{\operatorname{P-SSD}}, \, \forall x \in \Mc,
  \end{align*}
  %
%  \marginJC{You mean $ \tilde{D}(x^+) = \tilde{D}(x)
%    K_{\operatorname{P-SSD}}$?}
%  \marginMH{Yes, thanks!}
%
%  or
%  \begin{align*}
%    \phib(x^+) = \Lambda \phib(x), \forall x \in \Mc ,
%  \end{align*}
%  with $\phib(x) = [\phi_1(x) , \ldots, \phi_7(x)]^T$ and $\Lambda =
%  \diag(\lambda_1, \ldots, \lambda_7)$ 
  to describe the behavior
  of~\eqref{eq:ex-nonlinear} in a linear way.
  
  { % begin box to localize effect of arraystretch change
    \renewcommand{\arraystretch}{1.3}
    \begin{table}[htb]
      \centering
      \caption{Identified eigenfunctions and eigenvalues of the
        Koopman operator associated with the dynamics~\eqref{eq:ex-nonlinear}.}\label{table:eigenfunctions} 
      \begin{tabular}[htb]{ | l | l | }
        \hline
        \textbf{Eigenfunction} & \textbf{Eigenvalue} \\ \hline 
        $\phi_1(x) = 1 $  & $\lambda_1 = 1$ \\ \hline 
        $\phi_2(x) = x_1 $  & $\lambda_2 = 1.2$ \\ \hline
        $\phi_3(x) = x_1^2$  & $\lambda_3= 1.44$ \\ \hline
        $\phi_4(x) = x_1^3$  & $\lambda_4= 1.728$ \\ \hline
        $\phi_5(x) = 2 \, x_2^3 - 25 \, x_1^2 - 1$  & $\lambda_5 = 0.8$ \\ \hline
        $\phi_6(x) = x_1^4 $  & $\lambda_6 = 2.0736$ \\ \hline
        $\phi_7(x) = 2 \, x_1 x_2^3 - 25 \, x_1^3 - x_1$  & $\lambda_7 = 0.96$ \\ \hline
      \end{tabular}
      \vspace*{-1.5ex}
    \end{table}
  } % end box

  \new{To demonstrate the effectiveness of our method in long-term
    predictions, following the Extended Dynamic Mode Decomposition
    (EDMD) method~\cite{MOW-IGK-CWR:15} and
    Remark~\ref{r:inv-subs-eigs-pssd}, and given an arbitrary
    dictionary $\Dc$, we define its linear prediction matrix as $K =
    \Dc(X)^\dagger \Dc(Y)$. We also define the following relative and
    angle error functions on a trajectory $\{x(k)\}_{k=0}^L$
%with length $L$
\begin{align}\label{eq:relative-angle-error}
E_{\text{relative}}(k) &= \frac{\big\| \Dc(x(k)) - \Dc(x_0)
	K^k \big\|_2}{ \| \Dc(x(k)) \|_2} \times 100,
\nonumber \\
E_{\text{angle}}(k) &= \angle \big(\Dc(x(k)),\Dc(x_0) K^k
\big). 
\end{align} 
We compare the prediction accuracy on the original dictionary $D$ with the dictionary $\tilde{D}$ identified by P-SSD on 1000 trajectories with length $L=15$. In order to make sure the trajectories remain in the
space on which we have trained our methods, we sample the initial
conditions from $ [-0.1,0.1] \times [-3,3] $.
}
%  To demonstrate the effectiveness of our method in long-term
%  predictions, we compare the P-SSD strategy with
%  EDMD~\cite{MOW-IGK-CWR:15} as performed on the original dictionary
%  $D$.  The EDMD matrix is calculated as $\Kedmd = \argmin_K \|D(Y) -
%  D(X) K \|_F^2 =\dx^\dagger \dy$.  To compare EDMD and P-SSD, we
%  define the following relative and angle prediction errors at time
%  step $k \in \naturals_0$ starting from the initial state $x_0$,
%   \begin{align}\label{eq:relative-angle-error}
%    E_{\text{relative}}^{\text{EDMD}}(k) &= \frac{\big\| D(x(k)) -
%      D(x(0)) \big( \Kedmd \big)^k \big\|_2}{ \| D(x(k)) \|_2}
%    \times 100 ,
%   \nonumber  \\
%    E_{\text{relative}}^{\text{P-SSD}}(k) &= \frac{\big\|
%      \tilde{D}(x(k)) - \tilde{D}(x(0)) \big( K_{\operatorname{P-SSD}} \big)^k
%      \big\|_2}{ \| \tilde{D}(x(k)) \|_2} \times 100 ,
%     \nonumber \\
%     E_{\text{angle}}^{\text{EDMD}}(k) &= \angle\Big(D(x(k)),D(x(0))
%    \big( \Kedmd \big)^k \Big) ,
%     \nonumber \\
%    E_{\text{angle}}^{\text{P-SSD}}(k) &= \angle\Big(\tilde{D}(x(k)),
%    \tilde{D}(x(0)) \big( K_{\operatorname{P-SSD}} \big)^k \Big) ,
%  \end{align}
%  where $x(k)$ is the state of the system at time $k \in \naturals_0$.
%  We compare the accuracy of EDMD and P-SSD on 1000 trajectories with
%  length $L=15$. In order to make sure the trajectories remain in the
%  space on which we have trained our methods, we sample the initial
%  conditions from $ [-0.1,0.1] \times [-3,3] $.
Figure~\ref{fig:prediction error} shows the \new{median} and range
between the first and third quartiles of the relative and angle
prediction errors for $D$ and $\tilde{D}$. The P-SSD method has zero
prediction error since it identifies and predicts the evolutions on
Koopman-invariant subspaces. In contrast, the prediction errors
\new{on the original dictionary} are significantly large, even for
short time steps. \oprocend
\end{example}

%\marginMH{Please read: I think it is better to remove the standard
%  deviation. The picture is a little bit misleading since it might
%  convey to the reader that the errors can be negative. Since the data
%  does not come from a normal distribution, I do not if the standard
%  deviation can convey the right information to the reader.}
%
%\marginJC{Masih, why don't you do a confidence interval instead tuning
%  the confidence level to make sure error is not negative?  Bars
%  should not to be symmetric around the average.}
%
%\marginMH{Done!}

\begin{figure}[htb]
  \centering 
  % \subfloat[]
  {\includegraphics[width=.45
    \linewidth]{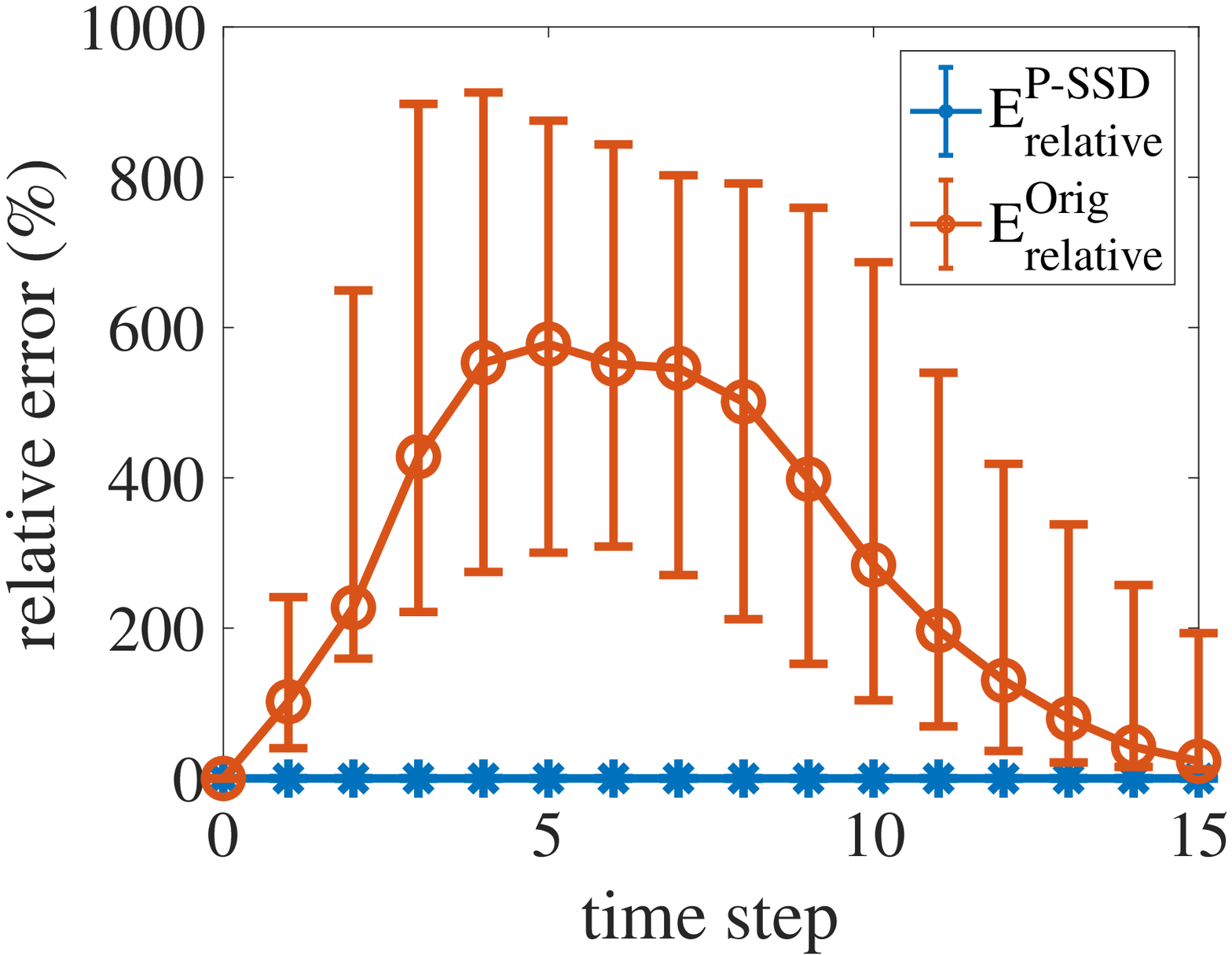}}
  % \subfloat[]
  {\includegraphics[width=.45
    \linewidth]{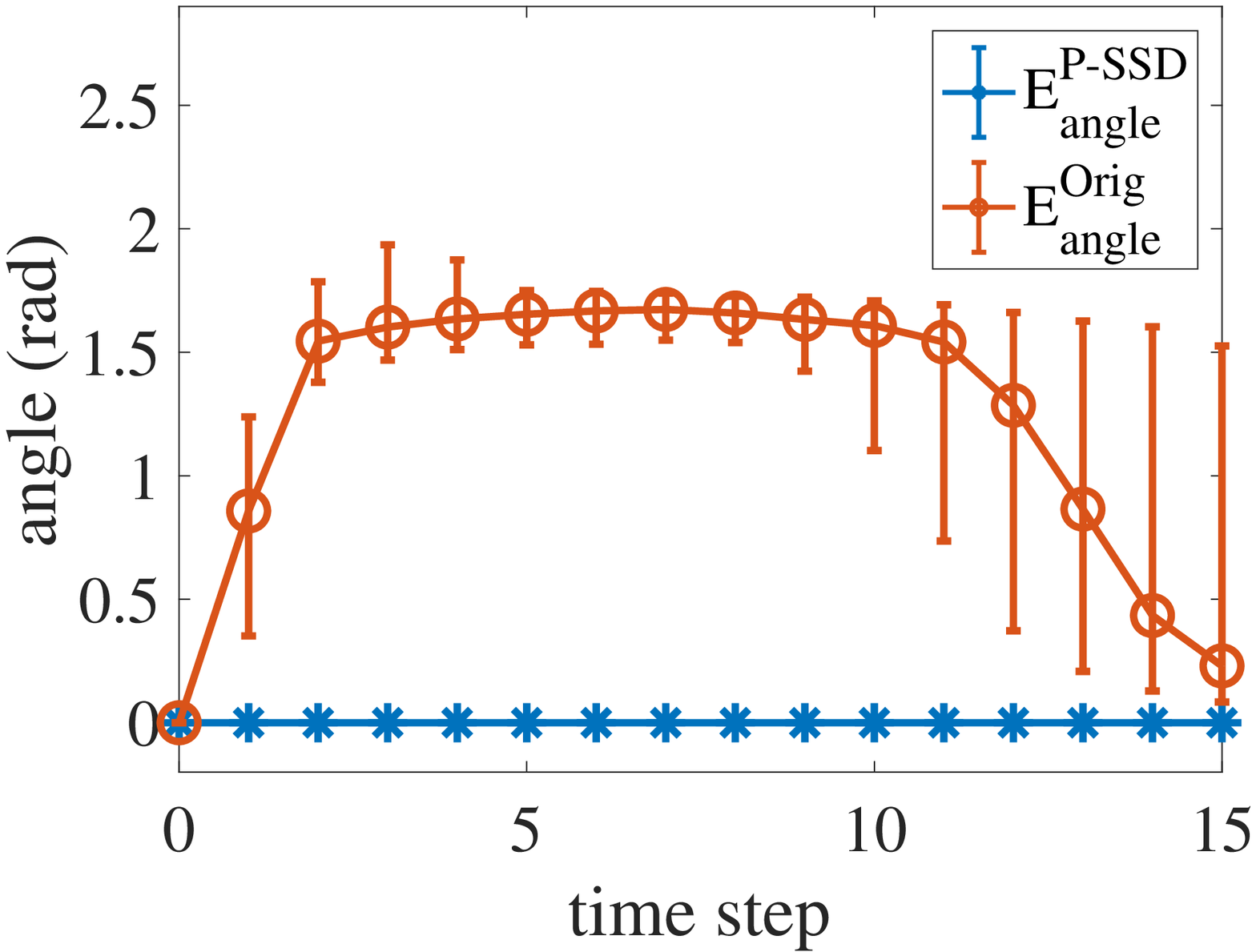}}
  % \vspace{0.5cm}
  \caption{\color{black} Median and range between first and third quartiles of relative
    (left) and angle (right) prediction errors for the original dictionary
    and the dictionary identified by P-SSD for system~\eqref{eq:ex-nonlinear} on 1000
    trajectories of length $L=15$ with initial conditions
    randomly selected from $ [-0.1,0.1] \times [-3,3]
    $. }\label{fig:prediction error}
  \vspace*{-1.5ex}
\end{figure}

%%%%%%%%%%%%%%%%%%
% The following example is removed for the reason of space.
%%%%%%%%%%%%%%%%%%

As Example~\ref{ex:nonlinear} illustrates, the P-SSD algorithm usually
reaches an equilibrium fast and consequently packet drops do not have
the opportunity to significantly affect the time to convergence. Also,
P-SSD may find the maximal Koopman-invariant subspace even if the
agents do not share the signature snapshots.  The next example is
selected with three goals in mind: (i) showing the importance of
signature snapshots, (ii) confirming the tightness of the bound for
time complexity in Theorem~\ref{t:consensus-complexity}, and (iii)
illustrating the robustness of the P-SSD algorithm against packet
drops.

\begin{example}\longthmtitle{Piecewise Linear System with Packet
    Drops}\label{ex:piecewise-linear}
  Given the state space $\Mc = [-1,1]^n$, define the sets
  $\{\Sc_{k}\}_{k=1}^n$ by
  \begin{align*}
    \Sc_{k} = \setdef{x \in \Mc}{(0 < x_k \leq 1) \land (-1 \leq x_j
      \leq 0, j \neq k)},
  \end{align*}
  where $x_i$ represents the $i$th element of $x$.  Consider the
  $n$-dimensional system $ x^+ = T(x)$, where
  \begin{align*}
    T(x) = \bigg \lbrace %
    \begin{array}{cl}%
      \left(I_n + (k^{-1} -1) e_k e_k^T \right)x & \text{if} \; x \in
      \Sc_{k}, \, k \in \until{n} ,
      \\ %
      x  & \text{if} \; x \in \Mc \setminus \bigcup_{k=1}^n \Sc_{k} .
    \end{array}  %
  \end{align*}
  Here, $e_k$ is the $k$th column of the identity matrix $I_n$.  We
  take $n =10$ and consider $M = 10$ agents connected according to a
  directed ring graph with diameter $d = 9$. We use the dictionary $D$
  comprised of all $N_d = 66$ monomials of the form $\prod_{i=1}^2
  y_i$, where $y_i \in \{x_1, \ldots, x_{10} \} \cup \{1\}$ for $i \in
  \{1,2\}$.
  
  To perform the simulations, we sample 100 data snapshots with
  initial conditions in $ \Mc \setminus \bigcup_{k=2}^n \Sc_{k}$ as
  our signature snapshots. In addition, we upload additional 1000 data
  snapshots sampled from $\Sc_{k}$ to the $k$th agent, for each $k \in
  \until{10}$. To check robustness against packet drops, we randomly
  drop each packet with $P$ percent chance, where we employ $P \in
  \{0,10,\ldots,90\}$. In all cases, the P-SSD algorithm correctly
  identifies the maximal Koopman-invariant subspace spanned by
  $\{1,x_1,x_1^2\}$.  Figure~\ref{fig:iterations-packet-drops} shows
  the average number of iterations (over 20 simulations) taken by
  P-SSD to achieve consensus. The plot shows how P-SSD achieves
  consensus relatively fast even in the presence of $90\%$ packet
  drops in a directed ring network.  The first column of
  Figure~\ref{fig:iterations-packet-drops} indicates that when there
  is no packet drops, P-SSD reaches consensus in 10 iterations, which
  is in agreement with the bound provided in
  Theorem~\ref{t:consensus-complexity}(a).

  \begin{figure}[htb]
    \centering 
    \includegraphics[width=.7 \linewidth]{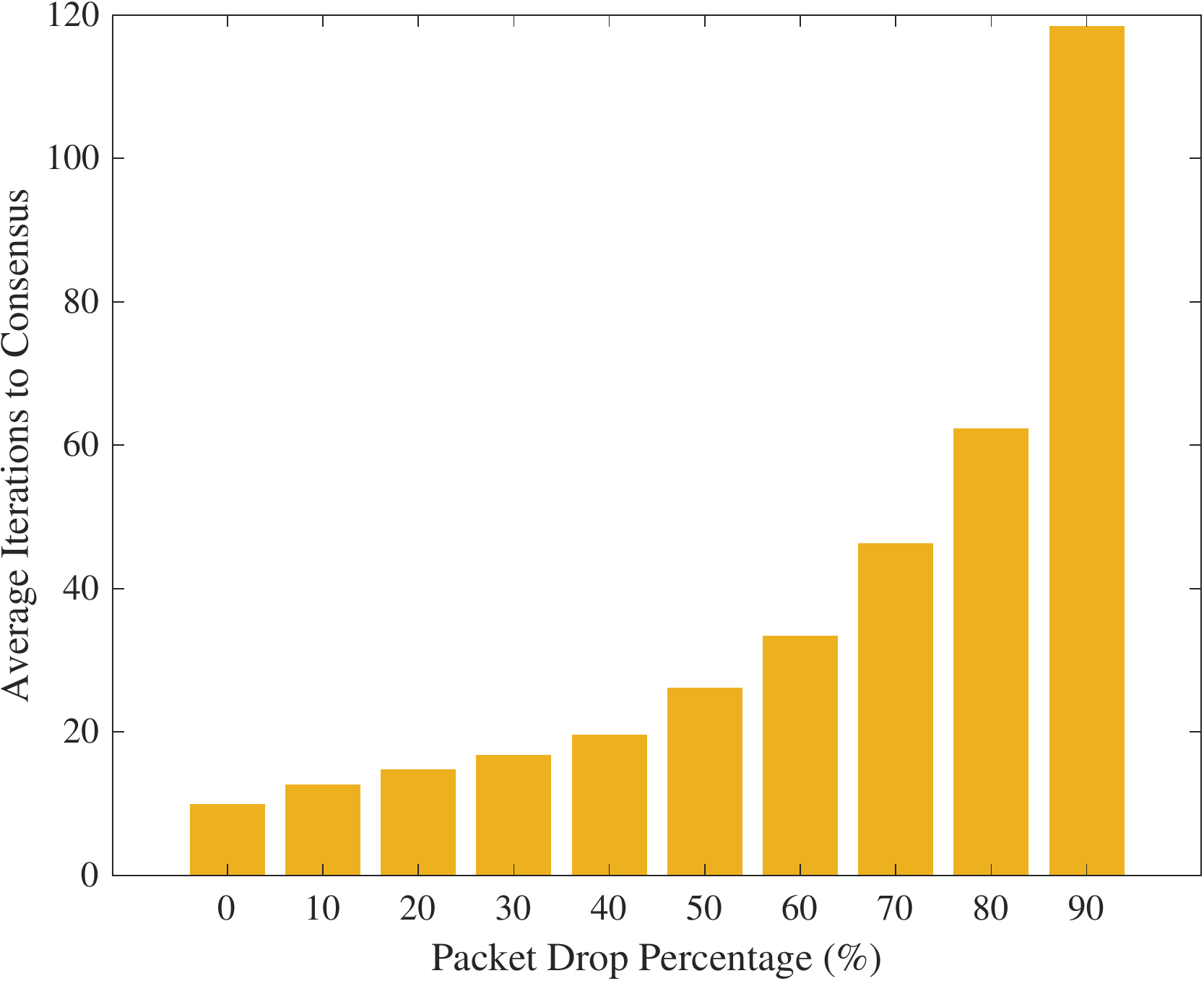}%
    \caption{Average iterations taken for P-SSD to achieve consensus
      versus packet drop percentage for
      Example~\ref{ex:piecewise-linear}.}\label{fig:iterations-packet-drops}
    \vspace*{-1.5ex}
  \end{figure}
  
  To illustrate the importance of the signature data snapshots, we
  also perform the simulations without them. In this case, the P-SSD
  algorithm fails, as it incorrectly identifies the whole space
  spanned by dictionary $D$ as a Koopman-invariant subspace. This
  happens because $\Span(D)$ is a Koopman-invariant subspace for the
  system with state space restricted to any of the subsets $\Sc_{k}$,
  $k \in \until{10}$. Moreover, these different Koopman operators have
  the same eigenfunctions. However, some of the corresponding
  eigenvalues are different for different $k$'s. Hence, those
  eigenfunctions for restricted systems are not eigenfunctions of the
  Koopman operator for the system defined on the whole state
  space~$\Mc$. The signature snapshots enable the agents to detect
  such inconsistencies and prevent the failure of the P-SSD algorithm.
  \oprocend
\end{example}

\begin{example}\longthmtitle{Van der Pol Oscillator}\label{ex:vanderpol}
  Here, we provide an example of the approximation of Koopman
  eigenfunctions and invariant subspaces when the original dictionary
  does not contain exact informative eigenfunctions.  Consider the Van
  der Pol oscillator with state $x = [x_1, x_2]^T$,
\begin{align}\label{eq:vanderpol}
\dot{x}_1 &= x_2
\nonumber \\
\dot{x}_2 &= -x_1 +(1-x_1^2)x_2.
\end{align}
To gather data, we run $10^3$ simulations for $5s$ with initial
conditions uniformly selected from $\Mc = [-4,4] \times [-4,4]$. We
sample the trajectory data with time step $\Delta t = 0.05s$ resulting
in $N = 10^5$ data snapshots. Moreover, we use the dictionary $D$ with
$N_d =45$ comprised of all distinct monomials up to degree 8 of the
form $\prod_{i=1}^8 y_i$, with $y_i \in \{1,x_1,x_2\}$ for $i \in
\until{8}$. To avoid numerical problems caused by finite-precision
errors, we form the matrix $[\dx^T,\dy^T]^T$ and scale each dictionary
function such that the norm of columns of the aforementioned matrix
become equal. Note that this scaling does not change the functional
space spanned by the dictionary since each function in the dictionary
is scaled by a nonzero number in $\real$. The communication network is
modeled by a complete digraph with $M = 20$ processors resembling a
GPU. In addition, we use 1000 snapshots randomly selected from our
database as the signature data snapshots and distribute the rest of
the data evenly among the processors.

Note that the dictionary $D$ only contains the trivial Koopman
eigenfunction $\phi(x) \equiv 1$ with eigenvalue $\lambda =1$.
% which does not capture any information regarding the dynamics. 
Consequently, the P-SSD algorithm results in a one-dimensional
Koopman-invariant subspace providing exact prediction but no
information about the dynamics. To address this issue,
% without changing the dictionary,
we use instead the Approximated P-SSD algorithm presented in
Remark~\ref{r:approximated-PSSD} with $\epsilon = 0.005$. Moreover, we
set $\epsilon_{\cap}=0.005$ following
Remark~\ref{r:basis-intersection-implementation}. This algorithm
reaches a consensus equilibria after 4 iterations identifying a
4-dimensional subspace. It is worth mentioning that we could not
implement the Approximated SSD algorithm from~\cite{MH-JC:20-tac} on
the same dataset since it performs SVD on the whole dataset and
requires a large memory that is beyond our computational resources
(also, performing SVD on such large datasets may result in large
round-off errors and inaccurate results).

Using the output matrix of any of the agents,
cf. Remark~\ref{r:inv-subs-eigs-pssd}, we find the dictionary
$\tilde{D}$ with $\tilde{N}_d = 4$ functions and create the
approximated P-SSD matrix as
$K_{\operatorname{P-SSD}}^{\operatorname{apprx}} = \tdx^\dagger \tdy$.
%\begin{align*}
%K_{\operatorname{P-SSD}}^{\operatorname{apprx}} = \tdx^\dagger \tdy.
%\end{align*}
Moreover, we find the Koopman eigenfunctions approximated by the
eigendecomposition of
$K_{\operatorname{P-SSD}}^{\operatorname{apprx}}$. Approximated P-SSD
finds the only exact Koopman eigenfunction ($\phi(x) \equiv 1$ with
eigenvalue $\lambda = 1$) in the span of the original dictionary
correctly.  Figure~\ref{fig:vanderpol-eigenfunction} shows the leading
nontrivial approximated eigenfunction corresponding to the largest
eigenvalue, which captures the behavior of the Van der Pol oscillator.

\begin{figure}[htb]
  \centering 
  % \subfloat[]
  {\includegraphics[width=.46 \linewidth]{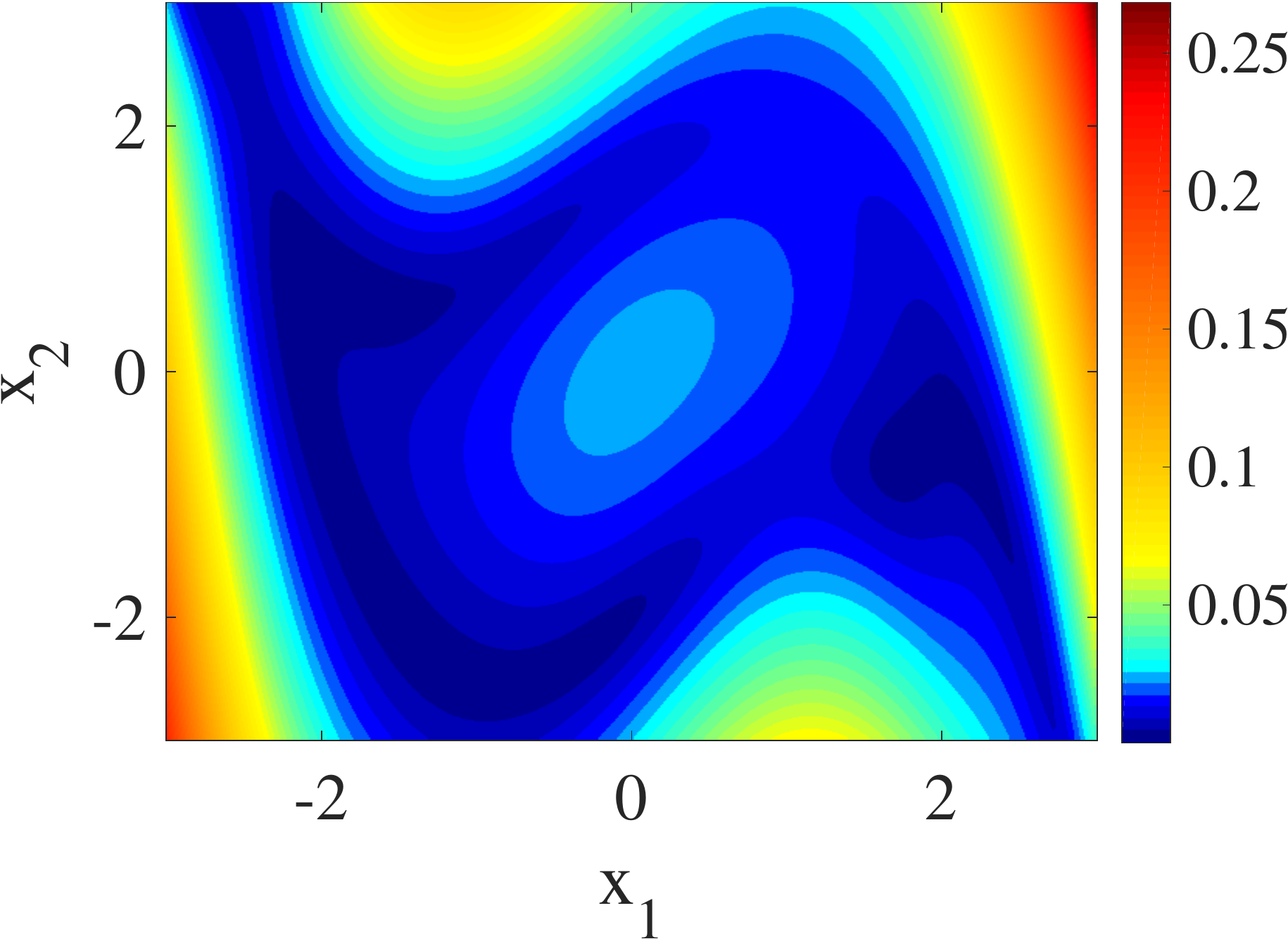}}
  % \subfloat[]
  {\includegraphics[width=.44	\linewidth]{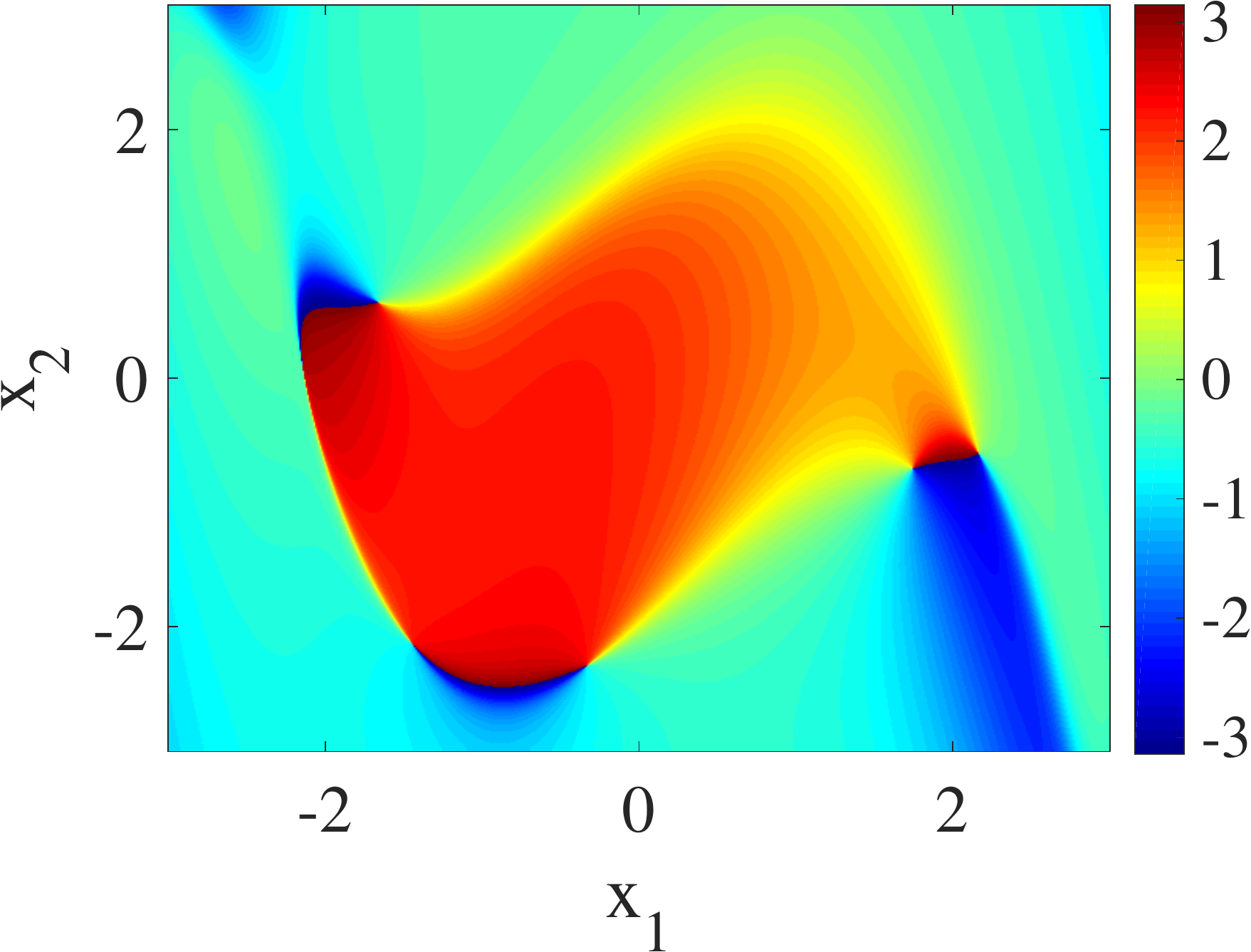}}
  \caption{{\color{black}Absolute value (left) and angle (right) of the
      approximated eigenfunction corresponding to eigenvalue $0.9647 +
      0.018j$.}}\label{fig:vanderpol-eigenfunction}
  \vspace*{-1.5ex}
\end{figure}

%To compare the effectiveness of our method against EDMD, we first note that the Root Relative Squared
%Error (RRSE) on the training data for prediction of $\dy$ with EDMD and $\tdy$ with Approximated P-SSD are $3.76\%$ and $0.4\%$ resp. Moreover, to 
\new{To show the advantage of our method in long-term
prediction,} we use the error functions presented
in~\eqref{eq:relative-angle-error} on 1000 trajectories with length $L
= 20$ time steps and initial conditions uniformly taken from $[-4,4]
\times [-4,4]$.
\begin{figure}[htb]
  \centering 
  % \subfloat[]
  {\includegraphics[width=.45 \linewidth]{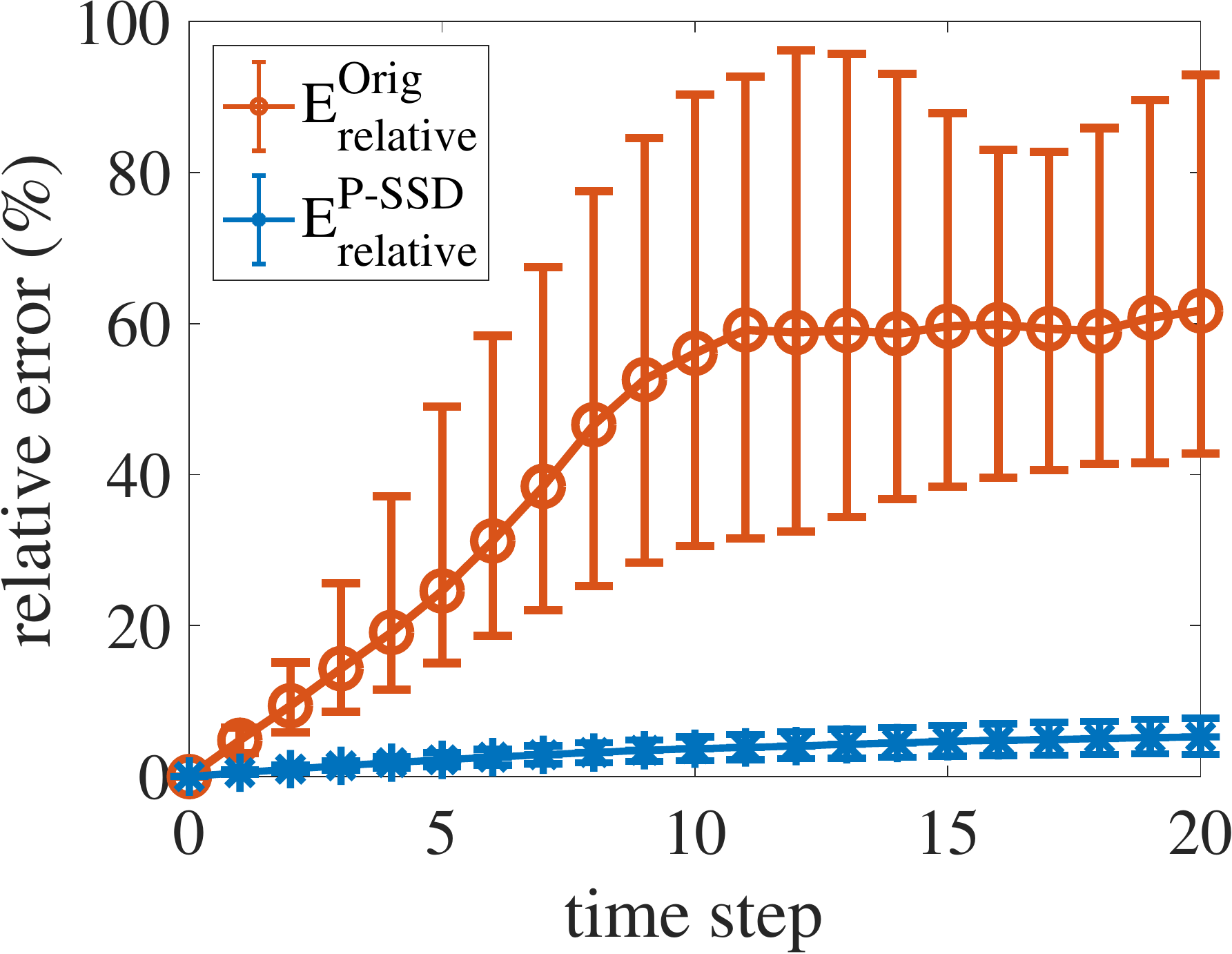}}
  % \subfloat[]
  {\includegraphics[width=.45	\linewidth]{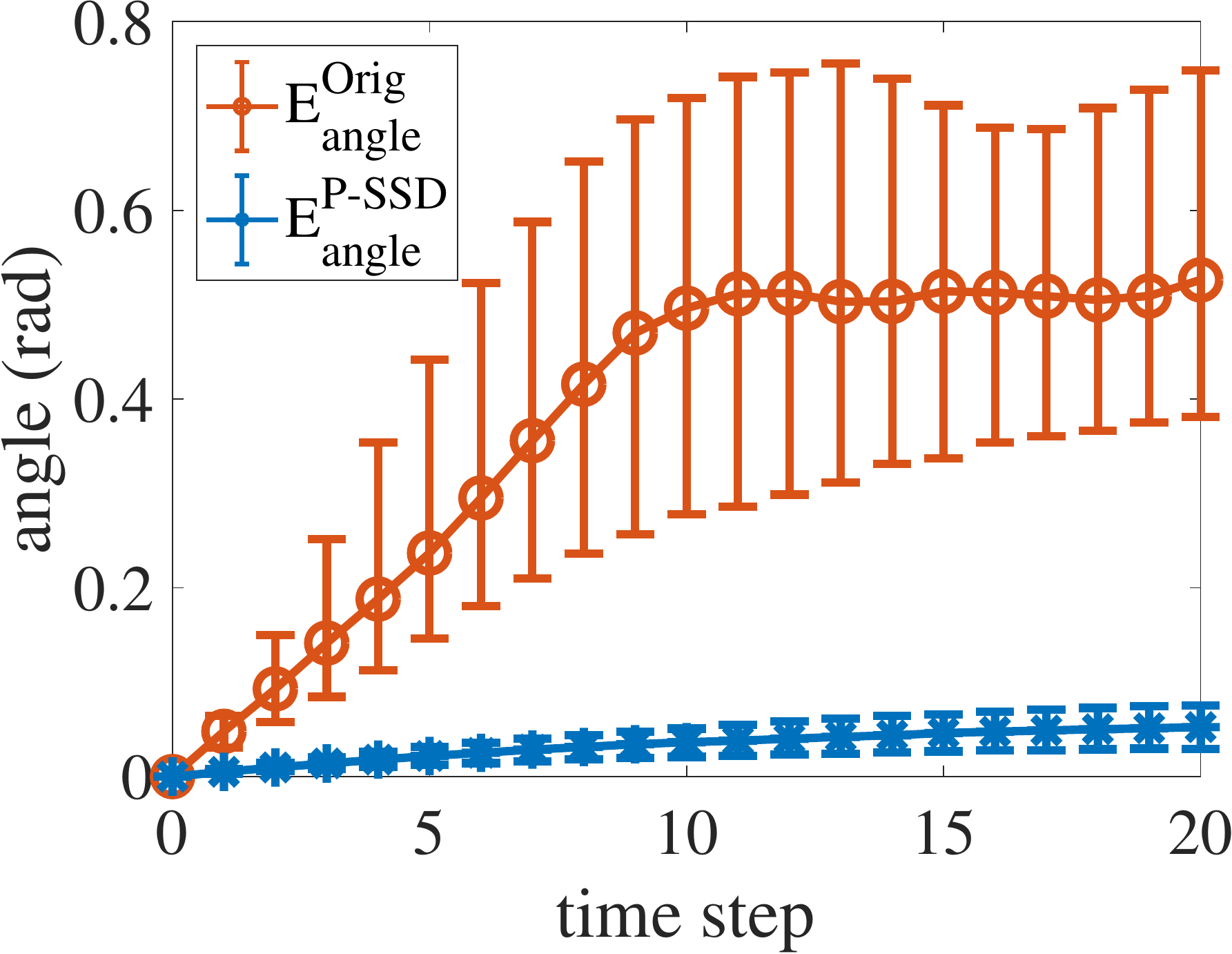}}
  \caption{\color{black}{Median and range between first and third quartiles of relative
      (left) and angle (right) prediction errors for the original dictionary 
      and the dictionary identified with Approximated P-SSD for system~\eqref{eq:vanderpol} on 1000 trajectories of length $L=20$ with initial conditions randomly selected from $\Mc$.}}\label{fig:vanderpol-longterm-error}
  \vspace*{-1.5ex}
\end{figure}
Figure~\ref{fig:vanderpol-longterm-error} shows the median and the
range between the first and third quartiles of the predictions errors
for \new{the dictionary $\tilde{D}$ identified by }Approximated P-SSD
and \new{the original dictionary $D$} over the aforementioned
trajectories. According to Figure~\ref{fig:vanderpol-longterm-error},
Approximated P-SSD has a clear advantage
%over EDMD 
in long-term prediction. After 20 time steps the median of the
\new{relative prediction errors for the original dictionary} is $60\%$
while the same error for \new{the dictionary identified with
  Approximated P-SSD} is $5\%$. Moreover, the median of the angle
errors after 20 time steps are $0.5$ and $0.05$ radians for \new{$D$}
and $\tilde{D}$, resp.  \oprocend
\end{example}

\new{
  \begin{example}\longthmtitle{Chaotic Lorenz System}\label{ex:lorenz}
    Consider the chaotic Lorenz system
    \begin{align}\label{eq:lorenz}
      \dot{x} &= 10(y-x)
      \nonumber \\
      \dot{y} &= x(28-z) - y
      \nonumber \\
      \dot{z} &= xy - (8/3) z,% -\frac{8}{3} z,
    \end{align}
    with state $s = [x, y, z]^T$ belonging to state space $\Mc =
    [-20,20] \times [-30,30] \times [0,50]$.  Our strategy for
    sampling from $\Mc$, the number of samples and signature data, and
    the processor network are similar to
    Example~\ref{ex:vanderpol}. We use the dictionary $D$ with $N_d
    =84$ comprised of all distinct monomials up to degree 6 of the
    form $\prod_{i=1}^6 y_i$, with $y_i \in \{1,x,y,z\}$ for $i \in
    \until{6}$. To avoid numerical problems caused by round-off
    errors, we scale the dictionary elements as in
    Example~\ref{ex:vanderpol}.  Note that the dictionary $D$ only
    contains the trivial Koopman eigenfunction $\phi(x) \equiv 1$ with
    eigenvalue $\lambda =1$.  Consequently, the P-SSD algorithm
    results in a one-dimensional Koopman-invariant subspace providing
    exact prediction but no information about the dynamics. To address
    this issue, we use instead the Approximated P-SSD algorithm
    presented in Remark~\ref{r:approximated-PSSD} with $\epsilon =
    0.001$. Moreover, we set $\epsilon_{\cap}=0.001$ following
    Remark~\ref{r:basis-intersection-implementation}. This algorithm
    reaches a consensus equilibria after 3 iterations identifying a
    2-dimensional subspace. It is worth mentioning that we could not
    implement the Approximated SSD algorithm from~\cite{MH-JC:20-tac}
    on the same dataset because of its large computational
    requirements..
    % since it performs SVD on the whole dataset and requires a large
    % memory that is beyond our computational resources (also,
    % performing SVD on such large datasets may result in large
    % round-off errors and inaccurate results).

    Using the output matrix of any of the agents,
    cf. Remark~\ref{r:inv-subs-eigs-pssd}, we find the dictionary
    $\tilde{D}$ with $\tilde{N}_d = 2$ and two approximated
    eigenfunctions in its span.
    % functions and create the approximated P-SSD matrix as
    % $K_{\operatorname{P-SSD}}^{\operatorname{apprx}} = \tdx^\dagger
    % \tdy$.  Moreover, we find the Koopman eigenfunctions
    % approximated by the eigendecomposition of
    % $K_{\operatorname{P-SSD}}^{\operatorname{apprx}}$.
    Approximated P-SSD finds the only exact Koopman eigenfunction
    ($\phi(x) \equiv 1$ with eigenvalue $\lambda = 1$) in the span of
    the original dictionary correctly. It also approximates another
    real-valued eigenfunction with eigenvalue $\lambda \approx 0.46$ which
    predicts that the trajectories converge to an \emph{approximated}
    invariant set since $|\lambda|<1$.

    To show the advantage of our method in long-term prediction, we
    use the error functions presented
    in~\eqref{eq:relative-angle-error} on 1000 trajectories with
    length $L = 20$ time steps and initial conditions uniformly taken
    from $\Mc = [-20,20] \times [-3,30] \times [0,50]$ and compare the
    prediction accuracy of the dictionary $\tilde{D}$ identified by
    Approximated P-SSD versus the original dictionary $D$.
    Figure~\ref{fig:lorenz-longterm-error} shows the median and the
    range between the first and third quartiles of the aforementioned
    predictions errors, indicating the advantage of P-SSD in long-term
    prediction.  \oprocend
\end{example}

\begin{figure}[htb]
	\centering 
	% \subfloat[]
	{\includegraphics[width=.45 \linewidth]{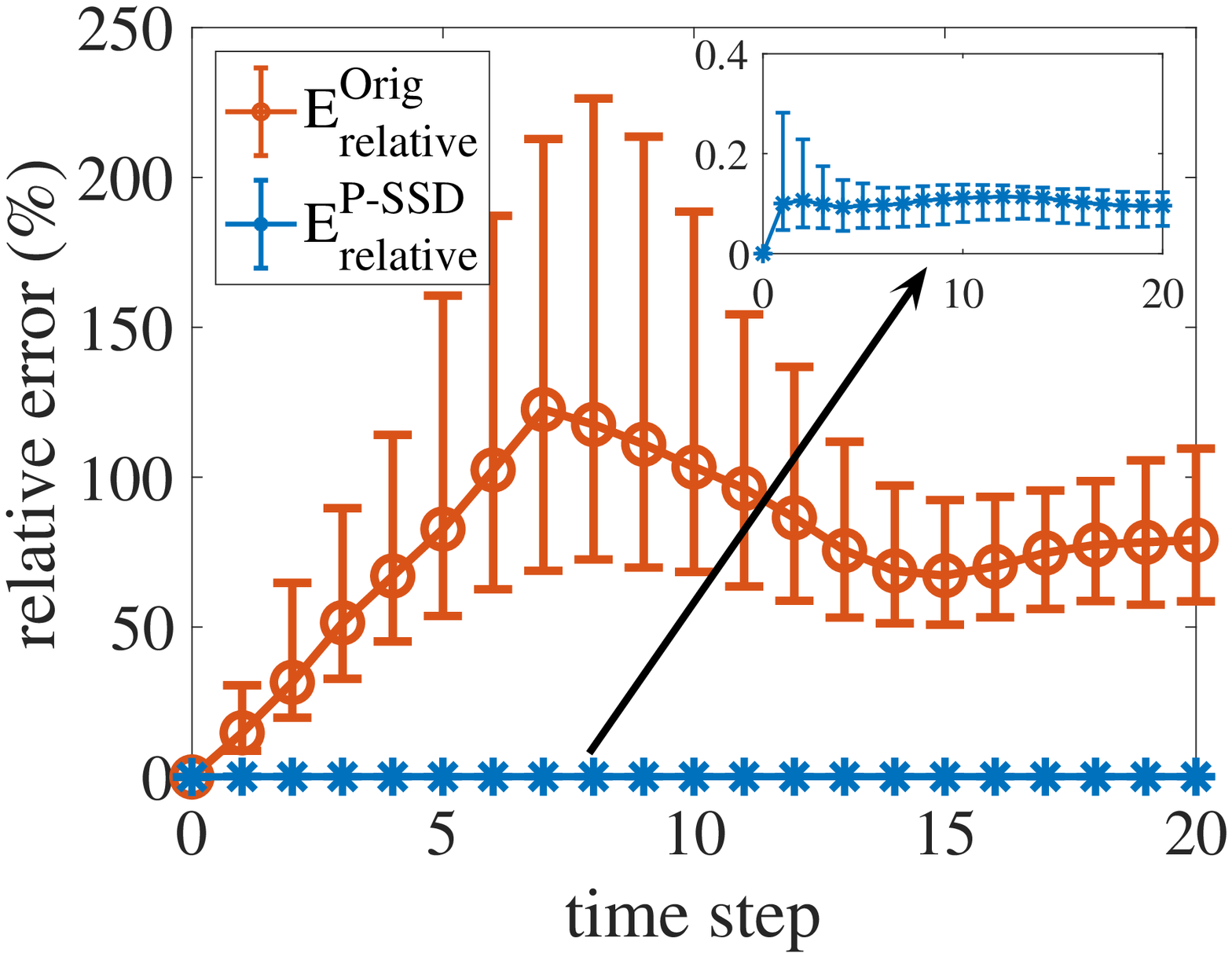}}
	% \subfloat[]
	{\includegraphics[width=.45	\linewidth]{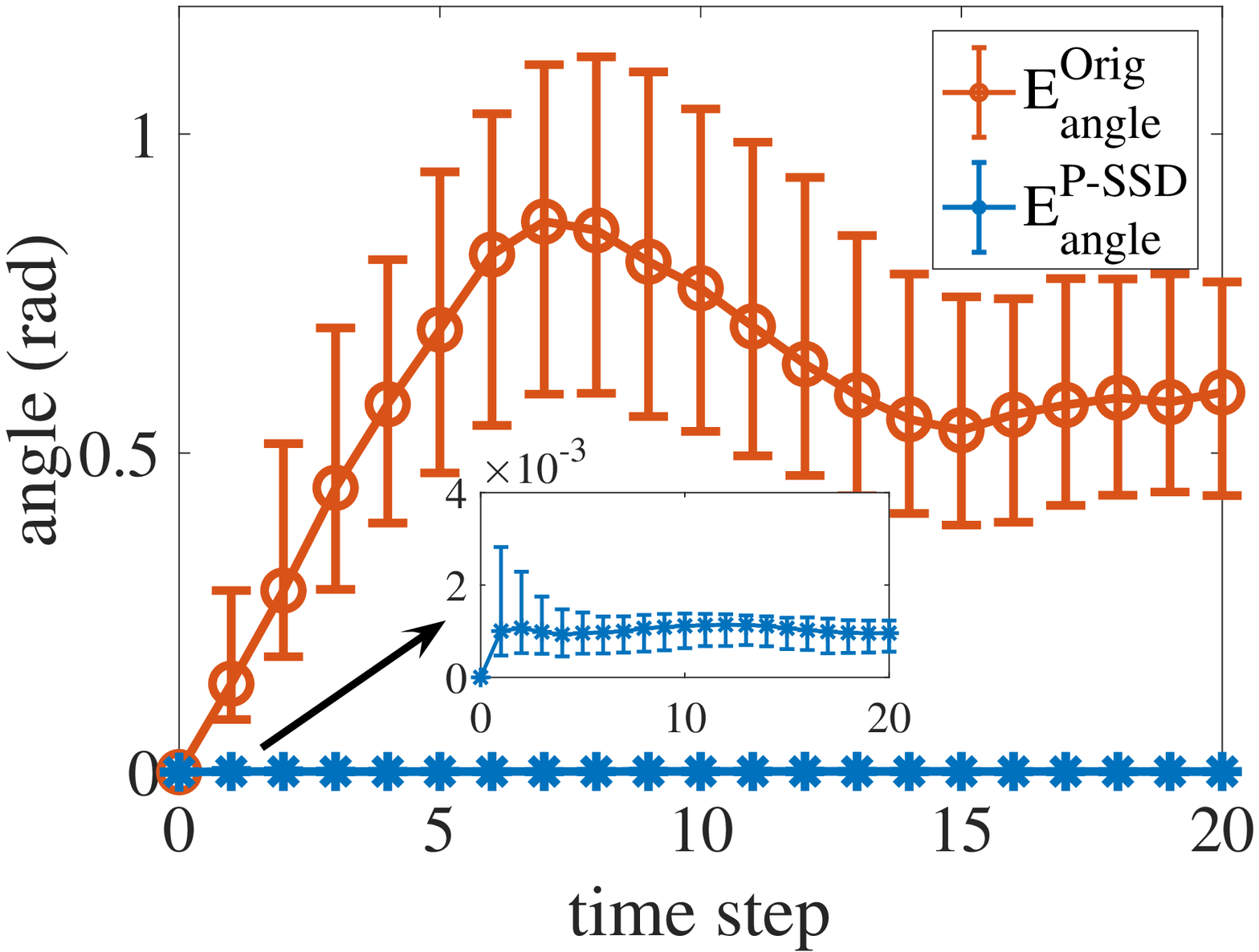}}
	\caption{\color{black}{Median and range between first and third quartiles of relative (left) and angle (right) prediction errors for dictionary $\tilde{D}$ identified by  Approximated P-SSD  and the original dictionary $D$ for system~\eqref{eq:lorenz} on 1000 trajectories of length $L=20$ with initial conditions randomly selected from $\Mc$.}}\label{fig:lorenz-longterm-error}
	\vspace*{-1.5ex}
\end{figure}
}

\section{Conclusions}
We have proposed a data-driven strategy, termed P-SSD algorithm, to
identify Koopman eigenfunctions and invariant subspaces corresponding
to an unknown dynamics. The strategy runs in parallel over a network
of processors that communicate over a digraph.  We have thoroughly
characterized the algorithm convergence and complexity properties. In
particular, we have identified conditions on the network connectivity
that ensure that the P-SSD's output coincide with that of the SSD
algorithm (a centralized method that under mild conditions on
  data sampling provably identifies the maximal Koopman-invariant
subspace in an arbitrary finite-dimensional functional space) when run
over the whole set of data snapshots.  The parallel nature of P-SSD
makes it run significantly faster than its centralized counterpart. We
have also established the robustness of P-SSD against communication
failures and packet drops.  Future work will explore to what extent
the data available to the individual agents and its density can be
employed in lieu of the common signature data, the development of
noise-resilient counterparts of the proposed algorithms, the synthesis
of distributed strategies that can work with streaming data sets,
  and the construction of dictionaries containing informative Koopman
  eigenfunctions leading to the approximation of the dynamics with
  accuracy  guarantees.

\bibliographystyle{IEEEtran}
\bibliography{alias,JC,Main,Main-add}

\begin{thebibliography}{10}
\providecommand{\url}[1]{#1}
\csname url@samestyle\endcsname
\providecommand{\newblock}{\relax}
\providecommand{\bibinfo}[2]{#2}
\providecommand{\BIBentrySTDinterwordspacing}{\spaceskip=0pt\relax}
\providecommand{\BIBentryALTinterwordstretchfactor}{4}
\providecommand{\BIBentryALTinterwordspacing}{\spaceskip=\fontdimen2\font plus
\BIBentryALTinterwordstretchfactor\fontdimen3\font minus
  \fontdimen4\font\relax}
\providecommand{\BIBforeignlanguage}[2]{{%
\expandafter\ifx\csname l@#1\endcsname\relax
\typeout{** WARNING: IEEEtran.bst: No hyphenation pattern has been}%
\typeout{** loaded for the language `#1'. Using the pattern for}%
\typeout{** the default language instead.}%
\else
\language=\csname l@#1\endcsname
\fi
#2}}
\providecommand{\BIBdecl}{\relax}
\BIBdecl

\bibitem{MH-JC:20-acc}
M.~Haseli and J.~Cort\'es, ``Fast identification of {K}oopman-invariant
  subspaces: parallel symmetric subspace decomposition,'' in \emph{{A}merican
  {C}ontrol {C}onference}, Denver, CO, Jul. 2020, pp. 4545--4550.

\bibitem{BOK:31}
B.~O. Koopman, ``Hamiltonian systems and transformation in {H}ilbert space,''
  \emph{Proceedings of the National Academy of Sciences}, vol.~17, no.~5, pp.
  315--318, 1931.

\bibitem{BOK-JVN:32}
B.~O. Koopman and J.~V. Neumann, ``Dynamical systems of continuous spectra,''
  \emph{Proceedings of the National Academy of Sciences}, vol.~18, no.~3, pp.
  255--263, 1932.

\bibitem{IM:05}
I.~Mezi{\'c}, ``Spectral properties of dynamical systems, model reduction and
  decompositions,'' \emph{Nonlinear Dynamics}, vol.~41, no. 1-3, pp. 309--325,
  2005.

\bibitem{MB-RM-IM:12}
M.~Budi{\v{s}}i{\'c}, R.~Mohr, and I.~Mezi{\'c}, ``Applied {K}oopmanism,''
  \emph{Chaos}, vol.~22, no.~4, p. 047510, 2012.

\bibitem{AM-IM:16}
A.~Mauroy and I.~Mezi{\'c}, ``Global stability analysis using the
  eigenfunctions of the {K}oopman operator,'' \emph{IEEE Transactions on
  Automatic Control}, vol.~61, no.~11, pp. 3356--3369, 2016.

\bibitem{AM-JG:16}
A.~Mauroy and J.~Goncalves, ``Linear identification of nonlinear systems: A
  lifting technique based on the {K}oopman operator,'' in \emph{{IEEE} Conf.\
  on Decision and Control}, Las Vegas, NV, Dec. 2016, pp. 6500--6505.

\bibitem{SK-FN-PK-HW-IK-CS-FN}
S.~Klus, F.~N{\"u}ske, P.~Koltai, H.~Wu, I.~Kevrekidis, C.~Sch{\"u}tte, and
  F.~No{\'e}, ``Data-driven model reduction and transfer operator
  approximation,'' \emph{Journal of Nonlinear Science}, vol.~28, no.~3, pp.
  985--1010, 2018.

\bibitem{AA-JNK:17}
A.~Alla and J.~N. Kutz, ``Nonlinear model order reduction via dynamic mode
  decomposition,'' \emph{SIAM Journal on Scientific Computing}, vol.~39, no.~5,
  pp. B778--B796, 2017.

\bibitem{CWR-IM-SB-PS-DSH:09}
C.~W. Rowley, I.~Mezi{\'c}, S.~Bagheri, P.~Schlatter, and D.~S. Henningson,
  ``Spectral analysis of nonlinear flows,'' \emph{Journal of Fluid Mechanics},
  vol. 641, pp. 115--127, 2009.

\bibitem{MK-IM-automatica:18}
M.~Korda and I.~Mezi{\'c}, ``Linear predictors for nonlinear dynamical systems:
  {K}oopman operator meets model predictive control,'' \emph{Automatica},
  vol.~93, pp. 149--160, 2018.

\bibitem{SP-SK:17}
S.~Peitz and S.~Klus, ``Koopman operator-based model reduction for
  switched-system control of {PDE}s,'' \emph{Automatica}, vol. 106, pp.
  184--191, 2019.

\bibitem{BH-XM-UV:18}
B.~Huang, X.~Ma, and U.~Vaidya, ``Feedback stabilization using {K}oopman
  operator,'' in \emph{{IEEE} Conf.\ on Decision and Control}, Miami Beach, FL,
  Dec. 2018, pp. 6434--6439.

\bibitem{EK-JNK-SLB:17}
E.~Kaiser, J.~N. Kutz, and S.~L. Brunton, ``Data-driven discovery of {K}oopman
  eigenfunctions for control,'' \emph{arXiv preprint arXiv:1707.01146}, 2017.

\bibitem{DG-DAP:17}
D.~Goswami and D.~A. Paley, ``Global bilinearization and controllability of
  control-affine nonlinear systems: a {K}oopman spectral approach,'' in
  \emph{{IEEE} Conf.\ on Decision and Control}.\hskip 1em plus 0.5em minus
  0.4em\relax IEEE, 2017, pp. 6107--6112.

\bibitem{GM-MC-XT-TM:19}
G.~Mamakoukas, M.~Castano, X.~Tan, and T.~Murphey, ``Local {K}oopman operators
  for data-driven control of robotic systems,'' in \emph{Robotics: Science and
  Systems}, Freiburg, Germany, Jun. 2019.

\bibitem{AN-JSK:20}
A.~Narasingam and J.~S. Kwon, ``Data-driven feedback stabilization of nonlinear
  systems: {K}oopman-based model predictive control,'' \emph{arXiv preprint
  arXiv:2005.09741}, 2020.

\bibitem{SHS-AN-JSK:20}
S.~H. Son, A.~Narasingam, and J.~S. Kwon, ``Handling plant-model mismatch in
  {K}oopman {L}yapunov-based model predictive control via offset-free control
  framework,'' \emph{arXiv preprint arXiv:2010.07239}, 2020.

\bibitem{AS-JE-AR-JPC-RMD:19}
A.~Salova, J.~Emenheiser, A.~Rupe, J.~P. Crutchfield, and R.~M. D’Souza,
  ``Koopman operator and its approximations for systems with symmetries,''
  \emph{Chaos}, vol.~29, no.~9, p. 093128, 2019.

\bibitem{AM-JB-MM:20}
A.~Mesbahi, J.~Bu, and M.~Mesbahi, ``Nonlinear observability via {K}oopman
  analysis: Characterizing the role of symmetry,'' \emph{Automatica}, p.
  109353, 2020.

\bibitem{PJS:10}
P.~J. Schmid, ``Dynamic mode decomposition of numerical and experimental
  data,'' \emph{Journal of Fluid Mechanics}, vol. 656, pp. 5--28, 2010.

\bibitem{JHT-CWR-DML-SLB-JNK:13}
J.~H. Tu, C.~W. Rowley, D.~M. Luchtenburg, S.~L. Brunton, and J.~N. Kutz, ``On
  dynamic mode decomposition: theory and applications,'' \emph{Journal of
  Computational Dynamics}, vol.~1, no.~2, pp. 391--421, 2014.

\bibitem{MOW-IGK-CWR:15}
M.~O. Williams, I.~G. Kevrekidis, and C.~W. Rowley, ``A data-driven
  approximation of the {K}oopman operator: Extending dynamic mode
  decomposition,'' \emph{Journal of Nonlinear Science}, vol.~25, no.~6, pp.
  1307--1346, 2015.

\bibitem{MK-IM:18}
M.~Korda and I.~Mezi{\'c}, ``On convergence of extended dynamic mode
  decomposition to the {K}oopman operator,'' \emph{Journal of Nonlinear
  Science}, vol.~28, no.~2, pp. 687--710, 2018.

\bibitem{SLB-BWB-JLP-JNK:16}
S.~L. Brunton, B.~W. Brunton, J.~L. Proctor, and J.~N. Kutz, ``{K}oopman
  invariant subspaces and finite linear representations of nonlinear dynamical
  systems for control,'' \emph{PLOS One}, vol.~11, no.~2, pp. 1--19, 2016.

\bibitem{MK-IM:19}
\BIBentryALTinterwordspacing
M.~Korda and I.~Mezi{\'c}, ``Optimal construction of {K}oopman eigenfunctions
  for prediction and control,'' 2019. [Online]. Available:
  \url{https://hal.archives-ouvertes.fr/hal-02278835}
\BIBentrySTDinterwordspacing

\bibitem{SP-NAM-KD:20}
S.~Pan, N.~Arnold-Medabalimi, and K.~Duraisamy, ``Sparsity-promoting algorithms
  for the discovery of informative {K}oopman invariant subspaces,'' \emph{arXiv
  preprint arXiv:2002.10637}, 2020.

\bibitem{MH-JC:20-tac}
M.~Haseli and J.~Cort\'es, ``Learning {K}oopman eigenfunctions and invariant
  subspaces from data: {S}ymmetric {S}ubspace {D}ecomposition,''
  \emph{\url{https://arxiv.org/abs/1909.01419}}, 2020.

\bibitem{QL-FD-EMB-IGK:17}
Q.~Li, F.~Dietrich, E.~M. Bollt, and I.~G. Kevrekidis, ``Extended dynamic mode
  decomposition with dictionary learning: A data-driven adaptive spectral
  decomposition of the {K}oopman operator,'' \emph{Chaos}, vol.~27, no.~10, p.
  103111, 2017.

\bibitem{NT-YK-TY:17}
N.~Takeishi, Y.~Kawahara, and T.~Yairi, ``Learning {K}oopman invariant
  subspaces for dynamic mode decomposition,'' in \emph{{C}onference on {N}eural
  {I}nformation {P}rocessing {S}ystems}, 2017, pp. 1130--1140.

\bibitem{EY-SK-NH:19}
E.~Yeung, S.~Kundu, and N.~Hodas, ``Learning deep neural network
  representations for {K}oopman operators of nonlinear dynamical systems,'' in
  \emph{{A}merican {C}ontrol {C}onference}, Philadelphia, PA, Jul. 2019, pp.
  4832--4839.

\bibitem{SEO-CWR:19}
S.~E. Otto and C.~W. Rowley, ``Linearly recurrent autoencoder networks for
  learning dynamics,'' \emph{SIAM Journal on Applied Dynamical Systems},
  vol.~18, no.~1, pp. 558--593, 2019.

\bibitem{FB-JC-SM:08cor}
F.~Bullo, J.~Cort{\'e}s, and S.~Martinez, \emph{Distributed Control of Robotic
  Networks}, ser. Applied Mathematics Series.\hskip 1em plus 0.5em minus
  0.4em\relax Princeton University Press, 2009.

\bibitem{JL-ASM-AN-TB:17}
J.~Liu, A.~S. Morse, A.~Nedi{\'c}, and T.~Ba{\c{s}}ar, ``Exponential
  convergence of a distributed algorithm for solving linear algebraic
  equations,'' \emph{Automatica}, vol.~83, pp. 37--46, 2017.

\bibitem{NAL:97}
N.~A. Lynch, \emph{Distributed Algorithms}.\hskip 1em plus 0.5em minus
  0.4em\relax Morgan Kaufmann, 1997.

\bibitem{DP:00}
D.~Peleg, \emph{Distributed Computing. A Locality-Sensitive Approach}, ser.
  Monographs on Discrete Mathematics and Applications.\hskip 1em plus 0.5em
  minus 0.4em\relax SIAM, 2000.

\bibitem{RMJ-PT:19}
R.~M. Jungers and P.~Tabuada, ``Non-local linearization of nonlinear
  differential equations via polyflows,'' in \emph{{A}merican {C}ontrol
  {C}onference}, Philadelphia, PA, 2019, pp. 1906--1911.

\bibitem{MH-JC:19-cdc}
M.~Haseli and J.~Cort\'es, ``Efficient identification of linear evolutions in
  nonlinear vector fields: {K}oopman invariant subspaces,'' in \emph{{IEEE}
  Conf.\ on Decision and Control}, Nice, France, Dec. 2019, pp. 1746--1751.

\end{thebibliography}

\appendices
\setcounter{equation}{0}
\renewcommand{\theequation}{\thesection.\arabic{equation}}
\section{Properties of the SSD Algorithm}\label{app:ssd}

Here, we gather some important results
from~\cite{MH-JC:19-cdc,MH-JC:20-tac} regarding the SSD algorithm and
its output.

% a centralized method to find eigenfunctions and finite-dimensional
% invariant subspaces of the Koopman operator associated with a
% dynamical system. We use the formulation in
% Section~\ref{sec:problem-statement} for dynamical
% system~\eqref{eq:dymamical-sys}, dictionary $D(x)$ and data snapshots
% $X,Y$.
	
\begin{theorem}\longthmtitle{Properties of the SSD's Output~\cite{MH-JC:20-tac}}\label{t:SSD-convergence}
  Let $ \cssd = \ssd (\dx,\dy)$ be the output of the SSD algorithm
  applied on $\dx, \dy$.  Given Assumption~\ref{a:full-rank},
  \begin{enumerate}
    % \item Stops after a finite number of iterations;
    % \marginMH{I will remove the first part if I don't use it
    % throughout the paper.}
  \item $\cssd$ is either $0$ or has full column rank and satisfies
    $\range(\dx \cssd) = \range(\dy \cssd)$;
  \item the subspace $\range(\dx \cssd)$ is maximal, in the sense
    that,
    % i.e.,
    for any matrix $E$ with $\range(\dx E) = \range(\dy E)$, we have
    $\range(\dx E) \subseteq \range(\dx \cssd)$ and $\range(E)
    \subseteq \range(\cssd)$.
  \end{enumerate}
\end{theorem}
\vspace*{1pt}

Theorem~\ref{t:SSD-convergence} provides the symmetric range equality
needed for identification of Koopman-invariant subspaces. Moreover, it
ensures the maximality of such symmetric subspaces. 
%  The next result
% shows that the identified subspaces by SSD remain monotone with
% respect to addition of data.

\begin{lemma}\longthmtitle{Monotonicity of SSD Output with Respect to
    Data Addition~\cite{MH-JC:20-tac}}\label{l:ssd-monotone}
  Let $D(X), D(Y)$ and $D(\hat{X}),D(\hat{Y})$ be two pairs of data
  snapshots such that
  \begin{align}\label{eq:data-monotone}
    \rows \big( [D(X), D(Y)] \big) \subseteq \rows
    \big([D(\hat{X}),D(\hat{Y})] \big),
  \end{align}
  and for which Assumption~\ref{a:full-rank} holds. Then
  \begin{align*}
    \range( \ssd([D(\hat{X}),D(\hat{Y})])) \subseteq \range(\ssd(D(X),
    D(Y))) .
  \end{align*}
\end{lemma}
\vspace*{1pt}

For $\cssd \neq 0$, we define the reduced dictionary
\begin{align}\label{eq:ssd-newdictionary}
  \tilde{D}(x) = D(x) \cssd, \quad \forall x \in \Mc.
\end{align}
Under some mild conditions on data sampling, the identified subspace
$\Span(\tilde{D}(x))$ converges to the \emph{maximal Koopman-invariant
  subspace} in $\Span(D(x))$ almost surely~\cite[Theorem
5.8]{MH-JC:20-tac} as the number of samples goes to infinity.

Based on Theorem~\ref{t:SSD-convergence}(a) and
Assumption~\ref{a:full-rank}, we deduce there exists a nonsingular
square matrix $\Kssd$ such that
\begin{align}\label{eq:ssd-newdictionary-data}
  \tilde{D}(Y) = \tilde{D}(X)\Kssd.
\end{align}
The eigendecomposition of $\Kssd$ specifies the functions in
$\Span(\tilde{D})$ that evolve linearly in time, i.e., given $\Kssd w
= \lambda w$ with $\lambda \in \cplx$ and $w \in \cplx^{\coln(\Kssd)}
\setminus \{0\}$, we have
\begin{align}\label{eq:linear-evolutions-in-new-dictionary}
  \tilde{D}(Y)w = \lambda \tilde{D}(X)w.
\end{align}
The next result % (a modified version of
% %~\cite[Theorem V.5]{MH-JC:19-cdc} and
% \cite[Theorem 4.5]{MH-JC:20-tac})
shows that the eigendecomposition of $\Kssd$ captures all the
functions in the span of the original dictionary that evolve linearly
in time according to the available data.

\begin{theorem}\longthmtitle{Identification of Linear Evolutions using
    the SSD Algorithm\cite{MH-JC:19-cdc}}\label{t:SSD-linear-evolutions}
  Under Assumption~\ref{a:full-rank}, $\Kssd w = \lambda w$ for some
  $\lambda \in \cplx$ and $w \in \cplx ^{\coln(\Kssd) \setminus
    \{0\}}$ iff there exists $v \in \cplx ^{N_d}$ such that $\dy v =
  \lambda \dx v$. In addition $v=\cssd w$.
\end{theorem}
\vspace*{1pt}

Theorem~\ref{t:SSD-linear-evolutions} shows that every function $\phi$
in the span of $D$ that satisfies $\phi(x^+) = \phi(T(x)) =\lambda
\phi(x)$ for every $x \in \rows(X)$ is in turn in the span of
$\tilde{D}$ and corresponds to an eigenvector of $\Kssd$.  This does
necessarily mean that $\phi$ is an eigenfunction of the Koopman
operator since its temporal linear evolution might not hold for all $x
\in \Mc$. Under reasonable assumptions on data sampling, the
identified functions are Koopman eigenfunctions almost
surely~\cite[Theorem 5.7]{MH-JC:20-tac}.

% \section{Algebraic Results}
% Here, we gather two linear algebraic results that we use frequently
% throughout the paper.

% \begin{lemma}\longthmtitle{Intersection of Linear
%     Spaces~\cite[Lemma A.1]{MH-JC:20-tac}}\label{l:subspace-intersection}
%   Let $A, B \in \real^{m \times n}$ with full column rank. Suppose the
%   columns of $Z=[(Z^A)^T,(Z^B)^T]^T \in \real^{2n \times l}$ form a
%   basis for the null space of $[A,B]$, with $Z^A,Z^B \in \real^{n
%     \times l}$.~Then,
%   \begin{enumerate}
%   \item $\range(AZ^A) = \range(A) \cap \range(B)$;
%   \item $Z^A$ and $Z^B$ have full column rank.
%   \end{enumerate}
% \end{lemma}

% \begin{lemma}[\cite{MH-JC:20-tac}]\label{l:product-subspace}
%   Let $A,C,D$ be matrices of appropriate sizes, with $A$ having full
%   column rank. Then $\range(AC) \subseteq \range(AD)$ if and only if
%   $\range(C) \subseteq \range(D)$.
% \end{lemma}

\vspace*{-3ex}

\begin{IEEEbiography}[{\includegraphics[width=1in,height=1.25in,clip,keepaspectratio]{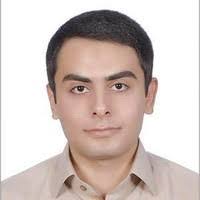}}]{Masih
    Haseli} was born in Kermanshah, Iran in 1991. He received the
  B.Sc. degree, in 2013, and M.Sc. degree, in 2015, both in Electrical
  Engineering from Amirkabir University of Technology (Tehran
  Polytechnic), Tehran, Iran. In 2017, he joined the University of
  California, San Diego to pursue the Ph.D. degree in Mechanical and
  Aerospace Engineering.  His research interests include system
  identification, nonlinear systems, network systems, data-driven
  modeling and control, and distributed computing.  Mr. Haseli was the
  recipient of the bronze medal in Iran's national mathematics
  competition in 2014.
\end{IEEEbiography}

\vspace*{-4ex}

\begin{IEEEbiography}[{\includegraphics[width=1in,height=1.25in,clip,keepaspectratio]{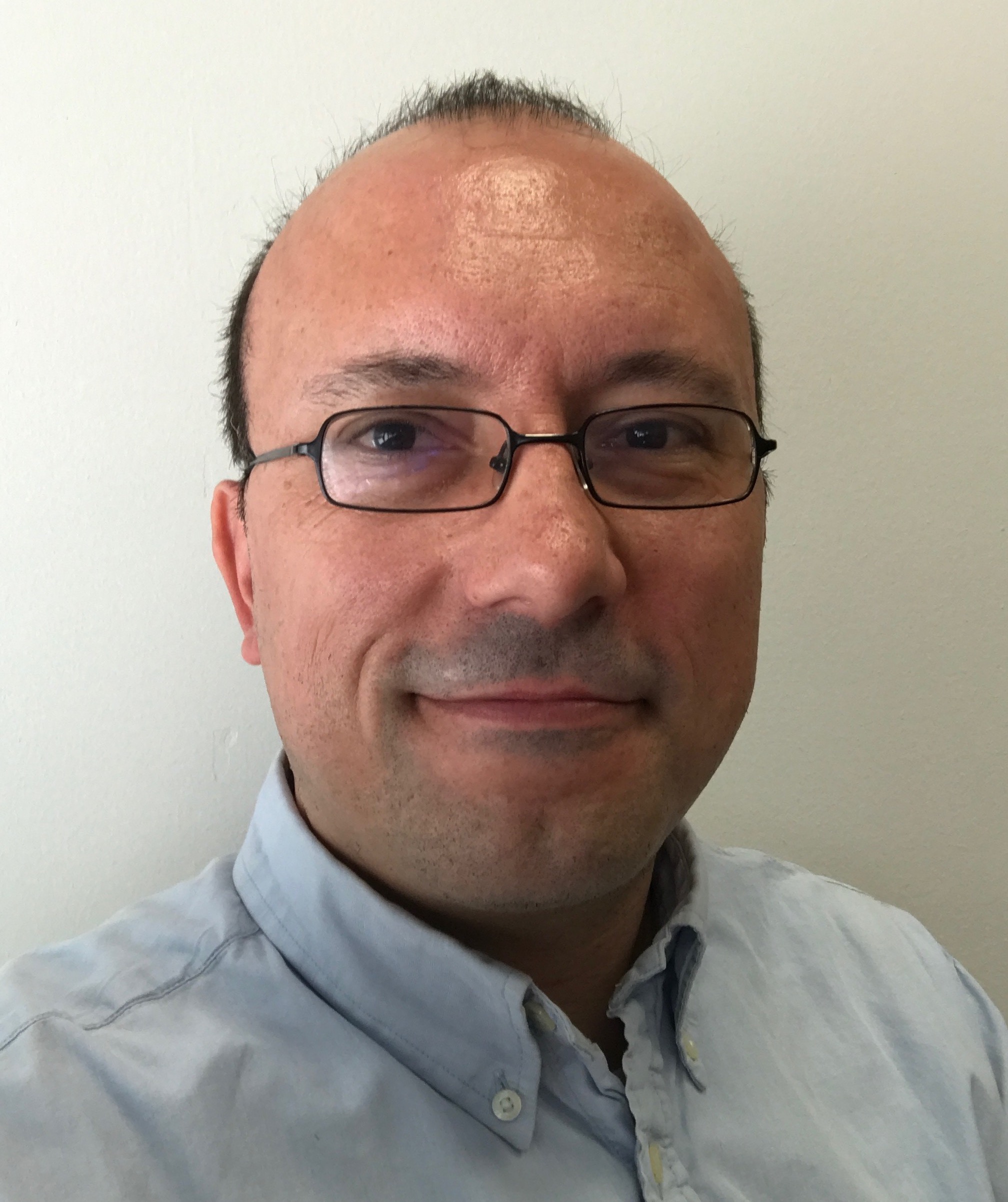}}]{Jorge
    Cort\'{e}s}
  (M'02, SM'06, F'14) received the Licenciatura degree in mathematics
  from Universidad de Zaragoza, Zaragoza, Spain, in 1997, and the
  Ph.D. degree in engineering mathematics from Universidad Carlos III
  de Madrid, Madrid, Spain, in 2001. He held postdoctoral positions
  with the University of Twente, Twente, The Netherlands, and the
  University of Illinois at Urbana-Champaign, Urbana, IL, USA. He was
  an Assistant Professor with the Department of Applied Mathematics
  and Statistics, University of California, Santa Cruz, CA, USA, from
  2004 to 2007. He is currently a Professor in the Department of
  Mechanical and Aerospace Engineering, University of California, San
  Diego, CA, USA.
  % He is the author of Geometric, Control and Numerical
  % Aspects of Nonholonomic Systems (Springer-Verlag, 2002) and
  % co-author (together with F. Bullo and S.  Mart{\'\i}nez) of
  % Distributed Control of Robotic Networks (Princeton University
  % Press,
  % 2009).
  He is a Fellow of IEEE and SIAM. 
  % At the IEEE Control Systems Society, he has been a Distinguished
  % Lecturer (2010-2014), and is currently its Director of Operations
  % and an elected member (2018-2020) of its Board of Governors.
  His current research interests include distributed control and
  optimization, network science, nonsmooth analysis, reasoning and
  decision making under uncertainty, network neuroscience, and
  multi-agent coordination in robotic, power, and transportation
  networks.
\end{IEEEbiography}

\end{document}